\documentclass[%
 reprint,
superscriptaddress,
nofootinbib,
 amsmath,amssymb,
 aps,
pra,
]{revtex4-2}
\usepackage{graphicx}
\usepackage{dcolumn}
\usepackage{bm}
\usepackage{hyperref}
\usepackage{blindtext}
\usepackage{braket}
\usepackage{amsthm}
\usepackage{booktabs}
\usepackage[english]{babel}
\usepackage[caption=false]{subfig}
\usepackage{orcidlink}
\makeatletter
\let\l@en\l@english
\makeatother
\usepackage{algorithm2e}
\SetKwComment{Comment}{\# }{}
\RestyleAlgo{ruled}

\newcommand\Wp{\ket{W_P}}

\DeclareMathOperator{\blockdiag}{blockdiag}
\newtheorem{theorem}{Proposition}
\newtheorem{corollary}{Corollary}[theorem]

\begin{document}

\preprint{APS/123-QED}

\title{Constrained Quantum Optimization via Iterative Warm-Start XY-Mixers}

\author{David Bucher~\orcidlink{0009-0002-0764-9606}}
\email{david.bucher@aqarios.com}
\affiliation{Aqarios GmbH, Munich, Germany}
\affiliation{LMU Munich, Department for Computer Science, Munich, Germany}

\author{Maximilian Janetschek}
\affiliation{Aqarios GmbH, Munich, Germany}

\author{Michael Poppel~\orcidlink{0009-0005-1141-0974}}
\affiliation{Aqarios GmbH, Munich, Germany}
\affiliation{LMU Munich, Department for Computer Science, Munich, Germany}

\author{Jonas Stein~\orcidlink{0000-0001-5727-9151}}
\affiliation{LMU Munich, Department for Computer Science, Munich, Germany}

\author{Claudia Linnhoff-Popien}
\affiliation{LMU Munich, Department for Computer Science, Munich, Germany}

\author{Sebastian Feld~\orcidlink{0000-0003-2782-1469}}
\email{s.feld@tudelft.nl}
\thanks{Corresponding author}
\affiliation{Delft University of Technology, Delft, The Netherlands}

\date{\today}

\begin{abstract}
The Quantum Approximate Optimization Algorithm (QAOA) is a leading hybrid heuristic for combinatorial optimization, but efficiently handling hard constraints remains a significant challenge. XY-mixers successfully confine quantum state evolution to a feasible subspace, such as the Hamming-weight-1 sector for one-hot constraints. On the contrary, warm-starting biases the search toward promising regions based on preliminary solutions. Combining these two techniques requires maintaining the essential alignment between the initial state and the mixer Hamiltonian to preserve convergence guarantees. Previous work demonstrated warm-starting with XY-mixers via a biased initial state, but relying only on standard mixer Hamiltonians. Consequently, the initial state is no longer a ground state of the mixer. In this work, we overcome these limitations by formulating a warm-started XY-mixer Hamiltonian for one-hot constraints and proving its ground-state properties. Furthermore, we provide a shallow circuit implementation suitable for NISQ implementations. We embed the warm-starting into a classical heuristic that iteratively updates the bias based on previous samples, called Iterative Warm-Starting (IWS). Extensive numerical simulations on Max-$k$-Cut and Traveling Salesperson Problem instances demonstrate that IWS-QAOA significantly accelerates the solution-finding process, increasing the probability of sampling optimal solutions by orders of magnitude compared to standard XY-QAOA. Finally, we validate our approach on the \texttt{ibm\_boston} QPU using hardware-tailored 144-qubit problem instances. By coupling IWS-QAOA with a greedy steepest-descent post-processing strategy to repair infeasible measurements caused by hardware noise, we successfully identify optimal solutions on actual quantum devices.
\end{abstract}

\maketitle

\section{Introduction}\label{sec:into}
The field of Quantum Computing (QC) has advanced significantly in recent years, transitioning from purely theoretical proposals~\cite{feynman1982} to tangible, Noisy Intermediate-Scale Quantum (NISQ) devices~\cite{preskill2018} and even to the first successful implementations of error-correction codes~\cite{acharya2025}. As hardware capabilities steadily improve and the scale of qubits with practically relevant fidelities grows~\cite{abughanem2025}, QC is becoming an increasingly more important option for complex computational challenges that stretch the limits of classical methods. Among the most promising real-world applications is combinatorial optimization, where quantum heuristics can explore vast solution spaces more efficiently and uncover high-quality solutions across domains such as logistics and energy~\cite{abbas2024, feld2019, krellner2025, blenninger2024a}.

The Quantum Approximate Optimization Algorithm (QAOA)~\cite{farhi2014a} has emerged as a prominent quantum heuristic for solving Combinatorial Optimization Problems (COPs), inspired by the adiabatic theorem~\cite{born1928}. Originally designed for Unconstrained Binary Optimization (UBO) problems, such as the well-known Max-Cut problem, QAOA alternates between the time evolution of a problem-specific cost Hamiltonian and a driver or mixer Hamiltonian. However, real-world combinatorial optimization problems frequently involve strict constraints. The standard approach of adding penalty terms to the objective function converts the constrained problem into an unconstrained one, but at the cost of significantly increasing the complexity of the search space.~\cite{lucas2014, glover2022}.

To efficiently address constrained optimization, recent research has focused on constraint-preserving mixers~\cite{hadfield2019, fuchs2022}. Most prominently, XY-mixers~\cite{wang2020} confine the quantum state evolution strictly to a feasible subspace, such as a specific Hamming-weight sector. This is particularly advantageous for implementing one-hot constraints (Hamming weight 1) without expanding the search space. Concurrently, \emph{warm-starting} techniques have been developed to bias the algorithm’s initial state toward promising regions, thereby accelerating convergence~\cite{egger2021, yu2022}. While warm-starting has been successfully applied to standard QAOA via relaxed classical solutions or adaptive bias fields, combining warm-starting with XY-mixers has so far only been demonstrated in a limited way: Ref.~\cite{carmo2025} proposed warm-started initial states, but relied on default XY-mixer Hamiltonians. Consequently, the initial state is no longer a ground state of the mixer Hamiltonian. Yet this property identifies QAOA with the adiabatic theorem, and breaking it undermines QAOA's performance guarantees for infinite depth $p\rightarrow \infty$.

In this work, we overcome these limitations by introducing a fully warm-started XY-mixer. Our contributions are multifold: first, we formulate the warm-started XY-mixer Hamiltonian and analytically prove its unique ground-state properties within the Hamming-weight 1 subspace. Second, we provide a shallow circuit implementation using two two-qubit Pauli rotations. Third, we apply the Iterative Warm Start (IWS) to QAOA, which avoids reliance on problem-specific classical solvers by iteratively updating a probability distribution using samples from previous iterations, akin to an adaptive bias field~\cite{yu2022, lopez-ruiz2025}, and test this scheme on an IBM Heron r3 Quantum Processing Unit QPU.

We benchmark the efficacy of our approach through extensive numerical simulations on the Max-$k$-Cut and Traveling Salesperson Problem. Our results show that IWS-QAOA significantly accelerates the solution-finding process, achieving higher approximation ratios and probabilities of sampling optimal solutions with fewer quantum resources. Finally, to validate our method on actual quantum hardware, we evaluate 144-qubit hardware-tailored, shallow-depth problem instances on the \texttt{ibm\_boston} QPU, using a greedy steepest-descent post-processing strategy to successfully mitigate constraint violations induced by hardware noise. We successfully identified optimal solutions in three of the five tested instances and solutions with an approximation ratio above 99\% for the other two. This marks one of the first large-scale applications of XY-mixers, embedding our work into the research on utility-scale quantum optimization~\cite{pelofske2023, mohseni2025, mohseni2026, romero2025}.

The remainder of this paper is structured as follows. Sec.~\ref{sec:backgroud} introduces the core concepts: QAOA, XY-Mixers, and warm-starting. In Sec.~\ref{sec:methods}, we propose our warm-started XY-mixer modification and present the IWS-QAOA algorithm. Sec.~\ref{sec:eval} presents the numerical simulation experiments and Sec.~\ref{sec:hw} showcases the results gathered from the NISQ QPU experiments. Finally, we conclude our findings in Sec.~\ref{sec:conclusion}.
\section{Background \& Related Work}\label{sec:backgroud}

\subsection{Quantum Approximate Optimization Algorithm}

The Quantum Approximate Optimization Algorithm (QAOA) was originally introduced as an approximation algorithm for the Max-Cut problem~\cite{farhi2014a}. Its scope was subsequently extended to a hybrid heuristic for general Ising-like spin-glass problems~\cite{wecker2016}, which are isomorphic to the Unconstrained Binary Optimization (UBO) problem class. In the quadratic case, this corresponds to the well-known QUBO format~\cite{lucas2014}.

QAOA can be interpreted as a parametrized $p$-step Trotterization of the adiabatic algorithm~\cite{blekos2024}. The process initializes in the equal superposition state $\ket{+}^{\otimes n}$, which is the ground state of the transverse-field mixer Hamiltonian $H_M = -\sum_j X_j$, where $j$ denotes the qubit index. Following the algorithm, we subsequently apply the time evolution of the cost Hamiltonian $H_C = \sum_x C(x) \ket{x}\bra{x}$, defined by the cost function $C:\{0,1\}^n\rightarrow \mathbb{R}$, and the mixer Hamiltonian in alternating order:
\begin{align}
    \ket{\vec\beta, \vec\gamma} = \prod_{i=0}^{p-1} e^{-i\beta_i H_M} e^{-i\gamma_i H_C} \ket{+}^{\otimes n}.
\end{align}
The variational parameters $\vec\beta$ and $\vec\gamma$ are optimized using a classical optimizer to minimize the quantum expectation value $C_Q: \mathbb{R}^{p} \times \mathbb{R}^p\rightarrow \mathbb{R}$, being now a continuous function:
\begin{align}
    C_Q(\vec\beta, \vec\gamma) = \bra{\vec\beta, \vec\gamma} H_C \ket{\vec\beta, \vec\gamma}.
\end{align}
Consequently, the probability of sampling an optimal solution $x_\text{opt}$, defined as $P_\text{opt} = \sum_{x_\text{opt}}|\braket{x_\text{opt}|\vec\beta,\vec\gamma}|^2$, where we sum over all degenerate optimal solutions, is expected to be significantly improved compared to random sampling ($P_\text{opt} \gg 2^{-n}$).

Reflecting the connection between QAOA and adiabatic evolution, recent studies have shown that a linear schedule---characterized by decreasing $\beta_i$ and increasing $\gamma_i$ values---performs remarkably well. This is particularly true for deeper circuits (large $p$), where linear parameters can circumvent the overhead of high-dimensional classical optimization~\cite{montanez-barrera2025, dehn2026}.

\subsection{Warm-Starting and Adaptive Bias}\label{sec:bg-ws}

Warm-starting is based on the idea of using classical computing to bias the quantum algorithm's search toward promising regions of the solution space. Ref.~\cite{egger2021} introduced warm-starting by biasing QAOA based on an input probability distribution $q_i$. This distribution is obtained either by solving a relaxed version of a binary optimization problem (i.e., by replacing binary variables with continuous ones) or by using a binary solution from a Semidefinite Programming (SDP) approximate algorithm, clipped to the interval $[\epsilon, 1-\epsilon]$ for $\epsilon > 0$. Consequently, the initial state for a single qubit $i$ is no longer an equal superposition but a biased rotation:
\begin{align}
    R_Y(\alpha_i)\ket{0} = \sqrt{q_i}\ket{0} + \sqrt{1-q_i} \ket{1},
\end{align}
where $\alpha_i = 2\arccos(\sqrt{q_i})$. The initial state for multiple qubits is given by the product of single-qubit states. To recover the ground-state property required by the adiabatic theorem, the mixer for a single qubit is aligned with the initial state~\cite{egger2021}:
\begin{align}
    \begin{aligned}
        H_M^{\text{WS}}(q_i) &= -X \sin \alpha_i - Z \cos \alpha_i \\
        &= \begin{pmatrix} 1 - 2q_i & -2\sqrt{q_i(1-q_i)} \\ -2\sqrt{q_i(1-q_i)} & 2q_i - 1 \end{pmatrix}. \label{eq:x-wsmixer}
    \end{aligned}
\end{align}
Following Ref.~\cite{egger2021}, similar warm-starting methods based on other classical heuristics and relaxations have emerged for Max-Cut-like optimization problems~\cite{tate2023, tate2023a}. Recently, warm-starting based on the solution to a regularized relaxed cost function has been proposed~\cite{he2026}.

Alternatively, Ref.~\cite{yu2022} introduced the adaptive bias QAOA, where a bias field $h_i$ is added to the mixer:
\begin{align}
    H_M^{\text{bf}}(h_i) = -X - h_i Z = \frac{1}{\sin\alpha_i} H^\text{WS}_M(q_i) \label{eq:x-bfmixer},
\end{align}
which is a scaled version of Eq.~\eqref{eq:x-wsmixer} with $h_i = 1 / \tan\alpha_i$. While both formulations are equivalent, adaptive bias differs by relying on an iterative update of the bias field based on samples from previous iterations, $h_i^{(t+1)} \gets h_{i}^{(t)}+\eta \langle Z_i\rangle^{(t)}$, where $\eta$ denotes the learning rate, rather than a fixed classical solution. This method was later combined with warm-starting from a classical solution in Ref.~\cite{yu2025}. Similar iterative approaches have also been discussed in Refs.~\cite{yuan2025, lopez-ruiz2025}.

The bias-field Digitized Counterdiabatic Quantum Optimization (bf-DCQO)~\cite{cadavid2025}, a QAOA variant derived via counterdiabatic driving with an adaptive bias field, follows a similar principle. The update rule for the adaptive bias field was further improved by evaluating $\langle Z_i \rangle$ with a conditioned value-at-risk (CVaR) metric~\cite{romero2025}. Notably, optimization with bf-DCQO has been successfully demonstrated on NISQ devices for large-scale optimization problems spanning 156 qubits~\cite{cadavid2025, romero2025}.

\subsection{XY-Mixers}

Standard QAOA is primarily designed for UBO problems, yet real-world combinatorial optimization problems (COPs) frequently involve hard constraints. Consequently, constraint-handling methods are essential to address a broader range of optimization problems. The conventional procedure involves adding penalty terms to the objective function, effectively transforming the COP into an unconstrained problem~\cite{glover2022, lucas2014}. However, this approach introduces significant drawbacks, such as increased search space complexity and an expansion of the total search space~\cite{bucher2025a}, even though constraints should ideally narrow it.

Since the mixer Hamiltonian in QAOA is responsible for exploring the search space, it is possible to construct mixers that restrict the quantum state evolution to a specific subspace of the entire Hilbert space, provided the initial state is prepared within that subspace~\cite{hadfield2019, fuchs2022}.

Most prominently, XY-mixers preserve the Hamming weight~\cite{wang2020} of a quantum state, making them ideal for implementing Hamming-weight constraints such as the one-hot constraint (Hamming weight 1), which is abundantly used in optimization modeling~\cite{gleixner2021}. For $k$ qubits encoding $k$ binary variables under a one-hot constraint, the fully-connected XY-mixer is defined as:
\begin{align}
\begin{aligned}
\mathcal{H} &= -\frac{1}{2(k-1)} \sum_{i > j}  (X_i X_j + Y_i Y_j) \\
&= -\frac{1}{k-1}\sum_{i > j} \mathrm{blockdiag}(0, X, 0)_{ij}, \label{eq:xy-mixer}
\end{aligned}
\end{align}
where the $(X_i X_j + Y_i Y_j)/2=\mathrm{blockdiag}(0, X, 0)$ block effectively corresponds to an ordinary $X$-mixer acting within the $\mathrm{span} \{\ket{01}, \ket{10}\}$ subspace. Beyond the fully-connected version, more hardware-efficient topologies have been proposed, such as the XY-ring-mixer, which restricts interactions to neighboring qubits~\cite{hadfield2019, wang2020}. For applying XY-mixers, the equal superposition initial state is replaced by $\ket{W}$ states, representing an equal superposition of all feasible one-hot solutions:
\begin{align}
    \ket{W} = \frac{1}{\sqrt{k}}\sum_i\ket{e_i},
\end{align}
where $\ket{e_i} = \ket{0\cdots010\cdots0}$ denotes the $i$-th basis state with the excitation at the $i$-th position. For problems involving multiple one-hot constraints, the algorithm employs a product of $\ket{W}$ states as the initial state, with corresponding XY-mixers for each constraint. Notably, XY-mixers cannot enforce constraints with shared binary variables.

While XY-mixers can be generalized to preserve other Hamming weights, our work focuses specifically on the Hamming-weight 1 case. Recent studies have demonstrated that XY-mixers overcome the fundamental limitations of penalty-based QAOA for one-hot-constrained problems, drastically reducing runtime~\cite{onah2025}. Consequently, they represent the tool of choice for integrating one-hot constraints into the QAOA framework.

\subsection{Warm-Starting and XY-Mixers}

Since the initial state of a QAOA utilizing XY-mixers must remain within a specific Hamming-weight sector, the single-qubit warm-starting scheme ~\cite{egger2021} described in Sec.~\ref{sec:bg-ws} cannot be applied directly. To address this, \citeauthor{carmo2025}~\cite{carmo2025} proposed a warm-started version of the $\ket{W}$-state:
\begin{align}
    \ket{W_P} = \sum_{i=1}^k \sqrt{P_i} \ket{e_i}, \label{eq:wp-state}
\end{align}
where the $P_i$ is the probability distribution over $k$ one-hot states that satisfies $\sum_i P_i = 1$~\cite{carmo2025} and $P_i > 0$. The $P_i = 0$ case can be covered by omitting the state $i$ entirely. $\Wp$-states can be implemented with linear circuit depth $O(k)$ on sparsely connected hardware and with $O(\log k)$ on fully-connected one. Details of the implementation are given in Appendix~\ref{app:wstate}. Similar to the approach in Ref.~\cite{egger2021}, the authors in Ref.~\cite{carmo2025} employ a rounded solution from an SDP relaxation, obtained from a penalized QUBO, to determine the initial probabilities $P$ and demonstrate performance improvements over standard XY-QAOA. However, a significant limitation remains: their approach does not modify the XY-mixer Hamiltonian, instead using the default, non-warm-started version~\cite{carmo2025}. This leaves the initial state misaligned with the mixer's ground state, a challenge we address in the following sections.

Finally, Ref.~\cite{kordonowy2026} approaches warm-starting with XY-mixers differently than warm-starting discussed so far. They first optimize a restricted subset of Lie-algebra generators in a multi-angle QAOA ansatz and then transfer these parameters to the full ansatz for fine-tuning. While they demonstrate promising results, their method differs fundamentally from warm-starting through biasing the search space via initial probabilities, as investigated in our work.
\section{Methods}\label{sec:methods}

It has been shown that QAOA performance is best when the initial state corresponds to the ground state of the mixer Hamiltonian~\cite{he2023}. Furthermore, aligning the initial state with the mixer facilitates the use of parameterizations inspired by the adiabatic algorithm, such as linear ramps of decreasing $\beta$ and increasing $\gamma$~\cite{egger2021}. While previous research on warm-starting QAOA with XY-mixers proposed a biased initial state, it did not provide a corresponding adapted mixer operator~\cite{carmo2025}. 

In this section, we first demonstrate how to warm-start the XY-mixer to maintain this necessary alignment in Sec.~\ref{sec:ws-mixers}. Next, we derive a shallow circuit implementation of the proposed mixer in Sec.~\ref{sec:circuit-impl}. Finally, Sec.~\ref{sec:iws-alg} describes the Iterative Warm Start (IWS) hybrid algorithm, a method inspired by adaptive-bias QAOA techniques~\cite{yu2025}.

\subsection{Warm-starting XY-Mixers}\label{sec:ws-mixers}

First, we show that $\Wp$ is \emph{not} a ground state of the (fully-connected) XY-mixer
Hamiltonian $\mathcal{H}$ defined in Eq.~\eqref{eq:xy-mixer}:
\begin{align}
\mathcal{H}\Wp &= -\frac{1}{k-1}\sum_i \sqrt{P_i} \sum_{j\neq i} \ket{e_j} \\
&=- \sum_i \frac{\sum_{j \neq i} \sqrt{P_j}}{k-1} \ket{e_i}.
\end{align}
This expression only reduces to the eigenvalue equation $\mathcal{H}\Wp = -\Wp$ if $P_i = 1/k \,\forall i$, which corresponds to the standard equal-superposition $\ket{W}$-state.

Following the logic of the warm-started X-mixer in~\eqref{eq:x-wsmixer}, we define the warm-started XY-mixer as
\begin{align}
   \mathcal{H}_{ij}(q) = \blockdiag(0, H^\text{WS}_M(q), 0)_{ij},
\end{align}
where the single-qubit warm-started mixer is embedded into the $\mathrm{span}\{\ket{01}, \ket{10}\}$ subspace of qubits $i$ and $j$. To distinguish these operators, we use calligraphic $\mathcal{H}$ for XY-type mixers and $H$ for standard X-mixers. From the definition of $H^\text{WS}_M$, it follows that the ground state of $\mathcal{H}_{ij}(q)$ is $\sqrt{q}\ket{01} + \sqrt{1-q}\ket{10}$.

\begin{theorem}\label{th:Hp-Wp-groundstate} The warm-started $k$-qubit $\Wp$-state with $P_i > 0$ is the unique ground state of
\begin{align}
    \mathcal{H}_P=\frac{1}{k-1}\sum_{i = 1}^{k}\sum_{j > i}^{k} \mathcal{H}_{ij}\left(q_{ij}\right),\,\quad q_{ij} = \frac{P_i}{P_i + P_j}
\end{align}
within the Hamming-weight $1$ subspace, with a corresponding energy of $-1$.
\end{theorem}

\begin{proof}
The proof proceeds in three steps: first, we demonstrate that $\ket{W_P}$ is an eigenstate of $\mathcal{H}_P$ with eigenvalue $-1$; second, we show that $\mathcal{H}_P$ leaves the Hamming-weight $1$ subspace invariant; and finally, we apply the Perron-Frobenius theorem to establish $\ket{W_P}$ as the unique ground state.

From the definition of $H^\text{WS}_M(q_{ij})$, we have:
\begin{align}
    H^\text{WS}_M(q_{ij}) = \frac{1}{P_i + P_j} \begin{pmatrix}
        P_j - P_i & -2 \sqrt{P_i P_j} \\ 
        -2 \sqrt{P_i P_j} & P_i - P_j
    \end{pmatrix},
\end{align}
which implies:
\begin{align*}
    \mathcal{H}_P \ket{e_i} = \frac{1}{k-1} \sum_{j \neq i} \left[ \frac{P_j - P_i}{P_i + P_j} \ket{e_i} - \frac{2 \sqrt{P_i P_j}}{P_i + P_j} \ket{e_j} \right].
\end{align*}
Expanding $\mathcal{H}_P \ket{W_P}$ using the definition of $\ket{W_P} = \sum_i \sqrt{P_i} \ket{e_i}$ yields:
\begin{align}
& \mathcal{H}_P \ket{W_P} \nonumber \\ &= \frac{1}{k-1} \sum_i \sum_{j \neq i} \left[ \sqrt{P_i} \frac{P_j - P_i}{P_i + P_j} \ket{e_i} - \frac{2 P_i \sqrt{P_j}}{P_i + P_j} \ket{e_j} \right] \nonumber \\ 
&= -\frac{1}{k-1} \sum_i \sqrt{P_i} (k - 1) \ket{e_i} = -\ket{W_P}.\label{eq:ws-proof-id}
\end{align}

To obtain the second equality in Eq.~\eqref{eq:ws-proof-id}, we first separated the two contributions in the double sum. In the second contribution, we swapped the labels of the summation indices $i$ and $j$, which turns the basis vector $\ket{e_j}$ into $\ket{e_i}$ and replaces the prefactor $P_i\sqrt{P_j}$ by $P_j\sqrt{P_i}$. Since the condition $j \neq i$ and the denominator $P_i+P_j$ are invariant under this relabeling, the summation domain remains unchanged. Both contributions can therefore be combined coefficient-wise for each basis state $\ket{e_i}$. The resulting numerator is $P_j-P_i-2P_j = -(P_i+P_j)$, so each summand reduces to $-\sqrt{P_i}\ket{e_i}$. Finally, for every fixed $i$, the inner sum contains exactly $k-1$ terms, which cancels the prefactor $1/(k-1)$ and yields $\mathcal{H}_P\ket{W_P}=-\ket{W_P}$. This confirms that $\ket{W_P}$ is an eigenstate with energy $-1$.

Next, we verify that $\mathcal{H}_P$ preserves the Hamming weight by showing that the commutator $[\mathcal{H}_{ij}(q), (I-Z)_i + (I-Z)_j] = 0$. Specifically:
\begin{align*}
    [\mathcal{H}_{ij}(q), (I-Z)_i] &= 4 \sqrt{q(1-q)}(\ket{10}\bra{01}- \ket{01}\bra{10})\\
    &= -[\mathcal{H}_{ij}(q), (I-Z)_j]. 
\end{align*}
Because the individual commutators cancel, it follows that $[\mathcal{H}_P, \mathcal{N}] = 0$, where $\mathcal{N} = \sum_i (I-Z)_i / 2$ is the number operator. Thus, any quantum state initialized in the Hamming-weight $1$ subspace remains within this feasible subspace under evolution of $\mathcal{H}_P$.

Finally, we observe that the off-diagonal matrix elements are negative for any two basis states, i.e., $\bra{e_i} \mathcal{H}_P \ket{e_j} < 0$ for all $i \neq j$. According to the Perron-Frobenius theorem~\cite{tasaki2020}, the unique ground state of such a matrix is the eigenvector whose entries are all real and positive. Since $\sqrt{P_i} > 0$ for all $i$, $\ket{W_P}$ satisfies this condition and is therefore the unique ground state in the Hamming-weight $1$ sector.
\end{proof}

\begin{corollary}\label{cor:1}
Let $G(V,E)$ be a mixer topology with $|V| = k$. If $G$ is connected and regular, then $\Wp$ with $P_i > 0$ is the unique ground state of
\begin{align}\label{eq:xy-ws-mixer-1}
    \mathcal{H}_P^G= \frac{1}{\Delta(G)}  \sum_{i,j\in E} \mathcal{H}_{ij}(q_{ij}),
\end{align}
within the Hamming-weight 1 subspace, with energy $-1$. 
\end{corollary}

Here, $\Delta(G) = \max\{\mathrm{deg}(v): v \in V\}$ denotes the maximum degree of the mixer topology. In Proposition~\ref{th:Hp-Wp-groundstate}, $G$ is a fully connected graph $K_k$ with $\Delta(K_k)=k-1$.

\begin{corollary}\label{cor:app}
Let $G(V,E)$ be a mixer topology with $|V| = k$. If $G$ is connected, then $\Wp$ with $P_i > 0$ is the unique ground state of
\begin{multline}\label{eq:xy-ws-mixer-2}
    \mathcal{H}_P^{G}= \frac{1}{\Delta(G)}\Bigg[\sum_{i,j\in E} \mathcal{H}_{ij}(q_{ij}) \\+  \sum_{i} [\deg(i)-\Delta(G)]\ket{e_i}\bra{e_i}\Bigg]
\end{multline}
within the Hamming-weight 1 subspace, with energy $-1$.
\end{corollary}

Note that Eqs.~\eqref{eq:xy-ws-mixer-1} and~\eqref{eq:xy-ws-mixer-2} are equivalent when $G$ is regular, since $\mathrm{deg}(v) = \Delta(G)\,\forall v\in V$.

\begin{corollary}\label{cor:2}
For any mixer topology $G(V,E)$ with $|V| = k$ and $E \neq \emptyset$, $\Wp$ with $P_i > 0$ is a ground state of the operator $\mathcal{H}_P^{G}$ defined in Eq.~\eqref{eq:xy-ws-mixer-2} within the Hamming-weight 1 subspace, with energy $-1$.
\end{corollary}
Detailed proofs for Corollaries~\ref{cor:1}, \ref{cor:app} and~\ref{cor:2} are provided in Appendix~\ref{sec:proofs-gs}.

\subsection{Circuit Implementation}\label{sec:circuit-impl}

First, we describe the implementation of the entire XY-mixer $\mathcal{H}_P^G$, assuming that the two-qubit warm-start XY-mixer block $e^{-i\beta \mathcal{H}_{ij}(q_{ij})}$ is available as a primitive. Subsequently, we provide the concrete gate decomposition for this block.

\subsubsection{Trotterization}

A valid mixer must facilitate transition probabilities between every pair of states within its domain~\cite{hadfield2019}. We therefore require the topology $G$ to be connected, ensuring that $\exists \beta : \bra{e_i} e^{-i\beta \mathcal{H}_P^{G}} \ket{e_j} \neq 0$ for all $i, j$. Under this condition, $\mathcal{H}_P^{G}$ (as defined in Eq.~\eqref{eq:xy-ws-mixer-2}) is a valid mixer with $\ket{W_P}$ as its unique ground state according to Corollary~\ref{cor:app}.

Because terms $\mathcal{H}_{ij}(q_{ij})$ and $\mathcal{H}_{kl}(q_{kl})$ do not commute if they share a qubit ($\{i,j\} \cap \{k,l\} \neq \emptyset$), we utilize Trotterization to implement the time evolution of $\mathcal{H}^G_P$. Since blocks $e^{-i\beta \mathcal{H}_{ij}(q_{ij})}$ acting on non-intersecting qubit pairs can be applied simultaneously, we first determine an edge coloring $\mathcal{G} = \{G_1, G_2, \dots\}$ of the graph $G$. Each subgraph $G_l$ is a collection of disjoint edges (a matching), such that $\bigcup_l G_l = G$ and $\Delta(G_l) = 1$. By Vizing's theorem~\cite{diestel2025}, the number of colors required is $|\mathcal{G}| \in \{\Delta(G), \Delta(G)+1\}$, with $\Delta(G)$ being the maximum degree in the graph $G$.

The time evolution for a single color $G_l$ is given by:
\begin{align*}
    e^{-i\beta\mathcal{H}_P^{G_l}} = \prod_{i,j\in E_i} e^{-i\beta \mathcal{H}_{ij}(q_{ij})} \prod_{i} e^{-i\beta (\deg (i) - 1) \ket{e_i}\bra{e_i}}.
\end{align*}
This is directly implementable, since the diagonal phase terms commute with the two-body XY terms. In this formulation, a relative phase is applied to qubits that do not share an edge within the current layer $G_l$ ($\deg(i) = 0$); otherwise, the exponent vanishes ($\deg(i) = 1$).

Consequently, the $T$-step Trotterization of $\mathcal{H}^G_P$ is given by:
\begin{align}\label{eq:trotter-t}
    e^{-i\beta\mathcal{H}_P^{G}} \approx \prod_{t=1}^T \left( e^{-i\beta \frac{|\mathcal{G}| - \Delta(G)}{T \Delta(G)}} \prod_{G'\in \mathcal{G}} e^{\frac{-i\beta}{T\Delta(G)}\mathcal{H}_P^{G'}} \right),
\end{align}
where the leading phase factor in each layer only applies if the number of colors is $|\mathcal{G}| = \Delta(G) + 1$, contributing a constant global phase to the expression. By including this term, we ensure that the mixer yields a consistent, topology-independent phase evolution on the ground state, $e^{-i\beta\mathcal{H}_P^{G}}\Wp = e^{i\beta}\Wp$, for any $T$.

Because QAOA can be interpreted as a Trotterized adiabatic evolution, and since $\Wp$ is a ground state of each individual sub-mixer $\mathcal{H}^{G'}_P$ (as established by Corollary~\ref{cor:2}), setting $T=1$ is equivalent to formulating a QAOA circuit with multiple sequential mixers rather than a single composite one. Given this equivalence and its hardware efficiency, we strictly focus on the $T=1$ case throughout the remainder of this manuscript.

\paragraph*{Example} 
The ring topology $G_\text{ring}$ is a common, hardware-friendly XY-mixer topology that is $2$-regular and connected, meaning $\Delta(G)=2$. An edge coloring partitions the edges into \emph{even} ($\{(0, 1),(2,3),\dots\}$) and \emph{odd} ($\{(1,2), (3,4), \dots\}$) subsets, along with a third, \emph{last} subset ($\{(0, k-1)\}$) if the number of qubits $k$ is odd~\cite{hadfield2019}. This yields the color set $\mathcal{G} =\{G_\text{even}, G_\text{odd}, G_\text{last}\}$, and the evolution is given by:
\begin{align*}
    e^{-i\beta\mathcal{H}_P^{G_\text{ring}}} = e^{-i\frac{\beta}{2}(\mathcal{H}_P^{G_\text{last}}+1) } e^{-i\frac{\beta}{2}\mathcal{H}_P^{G_\text{odd}}} e^{-i\frac{\beta}{2}\mathcal{H}_P^{G_\text{even}}},
\end{align*}
where the final $e^{-i\frac{\beta}{2}(\mathcal{H}_P^{G_\text{last}}+1)}$ is only applied when $k$ is odd.

For an even $k$, this implementation coincides with standard approaches reported in the literature, as the subgraphs $G_\text{even}$ and $G_\text{odd}$ are strictly 1-regular~\cite{hadfield2019,wang2020}. However, in the odd $k$ case, our formulation introduces additional phase factors to the idle qubits not involved in a layer, ensuring that every individual layer preserves $\Wp$ as an eigenstate.

\subsubsection{Implementation of the warm-start XY-block}

The circuit implementation for the time evolution of the single-qubit warm-start mixer $e^{-iH_M^\text{WS}(q)\beta}$ is given by the decomposition $R_Y(\alpha)R_Z(-2\beta)R_Y(-\alpha)$, where $\alpha = 2 \arccos\sqrt{q}$~\cite{egger2021}. This protocol can be extended to the XY-mixer case, which requires embedding these rotations into the single-excitation subspace via $\mathrm{blockdiag}(1,R_Y(\theta), 1)$ and $\mathrm{blockdiag}(1,R_Z(-2\beta), 1)$. These block-diagonal matrices can be decomposed into at most two-qubit Pauli gates as follows:
\begin{align}
    \mathrm{blockdiag}(1,R_Z(-2\beta), 1) &= R_Z(-\beta) \otimes R_Z(\beta), \\
    \mathrm{blockdiag}(1,R_Y(\theta), 1) &=  R_{YX}(\theta / 2)R_{XY}(-\theta / 2),\nonumber
\end{align}
where $R_{XY}(\varphi)=e^{-iXY\varphi/2}$, $R_{YX}(\varphi)=e^{-iYX\varphi/2}$.

However, this naive approach results in four consecutive two-qubit Pauli rotations. In contrast, the standard, non-warm-started XY-block $e^{-i \beta \mathcal{H}(1/2)} = U_{XY}(\beta) = R_{XX}(-\beta)R_{YY}(-\beta)$ can be implemented using only two two-qubit Pauli rotations\footnote{Highly optimized implementations exist for $U_{XY}(\beta)$ that require only two \textsc{CNOT} gates. See the \href{https://github.com/Qiskit/qiskit/blob/2.4.2/qiskit/circuit/library/standard_gates/xx_plus_yy.py\#L108-L121}{Qiskit \texttt{XXPlusYYGate} implementation (v2.4.2)}~\cite{javadi-abhari2024}, with $R_{XX+YY}(-2\beta) = U_{XY}(\beta)$}, where $R_{XX}(\varphi) = e^{-i XX \varphi / 2}$ and $R_{YY}(\varphi) = e^{-i YY \varphi / 2}$. Therefore, we propose an alternative, hardware-efficient decomposition of the evolution based strictly on $U_{XY}$.

\begin{theorem}\label{th:rotation}
The warm-started XY-block is given by the exact decomposition
\begin{align*}
    e^{-i\beta \mathcal{H}(q)} = (R_Z(\phi_1) \otimes I) U_{XY}(\phi_2) (I \otimes R_Z(-\phi_1)),
\end{align*}
with the angles defined as:
\begin{align}
\begin{aligned}
    \phi_1 &= \mathrm{arctan2}\left((1-2q)\sin\beta, \cos\beta\right), \\
    \phi_2 &= \arcsin\left( 2\sqrt{q(1-q)} \sin\beta\right).
\end{aligned}
\end{align}
\end{theorem}

The proof of Proposition~\ref{th:rotation} is provided in Appendix~\ref{sec:proof-rotation}.

\subsubsection{Scaling the XY-block}\label{sec:scaling}

Because the XY-part of $\mathcal{H}(q)$ diminishes as $q$ approaches the extreme points $q \to 0$ or $q \to 1$, the effective mixing magnitude $|\bra{01}e^{-i\beta \mathcal{H}(q)} \ket{10}|$ decreases. Consequently, we observe that the optimal $\beta$ values for QAOA increase as $\sqrt{q(1-q)}$ decreases. To counteract this and ensure consistent $\beta$ parameters, we implement a scaled and shifted version of $\mathcal{H}(q)$, defined as:
\begin{align}
    \tilde{\mathcal{H}}(q) = \frac{1}{2\sqrt{q(1-q)}} \left(\mathcal{H}(q) + I_{XY}\right) - I_{XY},
\end{align}
where $I_{XY} = \mathrm{blockdiag}(0, I, 0)$ and commutes with $\mathcal{H}(q)$. Since scaling and shifting do not alter the eigenstates, and the ground state energy remains $-1$ in this case, it follows that $\mathcal{H}_{ij}(q)\Wp = \tilde{\mathcal{H}}_{ij}(q)\Wp$. Therefore, we can safely replace $\mathcal{H}_{ij}(q)$ with $\tilde{\mathcal{H}}_{ij}(q)$ in all propositions and corollaries from Sec.~\ref{sec:ws-mixers}.

The circuit implementation of this modified time evolution simply requires rescaling the parameter $\beta \gets \beta / (2 \sqrt{q(1-q)})$ for the XY-block, alongside applying an additional phase to the $\{\ket{10}, \ket{01}\}$ sector. While this phase shift could be achieved with a two-qubit $R_{ZZ}$ gate, we restrict our focus to the Hamming-weight-1 subspace, meaning the Hamming-weight-2 sector ($\ket{11}$) is never populated. Thus, instead of a two-qubit gate, we can efficiently apply two single-qubit phase gates: $P(\varphi) \otimes P(\varphi)$ with
\begin{align*}
    \varphi = \left(1 - \frac{1}{2\sqrt{q(1-q)}}\right)\beta.
\end{align*}

\subsection{Iterative Warm-Starting Algorithm}\label{sec:iws-alg}

Having established a method to bias the XY-QAOA evolution towards an input probability distribution $P$, we now address how to determine $P$. As discussed in Sec.~\ref{sec:backgroud}, there are two primary methods for biasing search space exploration. First, warm-starting can be based on the solution of a classical solver (e.g., either a relaxed continuous solution or a rounded solution from an approximation algorithm)~\cite{egger2021, carmo2025}. Second, the adaptive bias approach iteratively updates a bias field based on intermediate solutions generated by the quantum algorithm, steering exploration toward previously discovered high-quality solutions~\cite{yu2022, yu2025}.

We adopt the latter, iterative approach. This avoids the need for a problem-specific classical algorithm to generate an initial solution. Furthermore, we observed that classical solvers default to exploring only integer solutions for relaxed versions of one-hot-constrained problems (optimal solutions lie at the vertices of polytopes~\cite{horst1996}), thereby limiting their utility as a source of bias in our case. The adaptive iterative approach, conversely, is universally applicable regardless of the structure of the input problem. While we use the mechanics of adaptive bias, we retain the \emph{probability} and \emph{warm-starting} terminology rather than the \emph{external-field} nomenclature, as it aligns more naturally with one-hot constraints.

Consider a binary optimization problem containing $L$ one-hot constraints, where each constraint $l$ spans $k_l$ non-overlapping binary variables. For simplicity, we assume every variable is associated with exactly one constraint; incorporating standalone unconstrained binary variables, when present in the input problem, is straightforward using standard methods~\cite{egger2021, yu2025}. By assigning a qubit to each binary variable and selecting a mixer topology $G_l$ for each constraint, we define the combined mixer as:
\begin{align}
    U_M(\beta, P) = \bigotimes_{l=1}^L e^{-i\beta \mathcal{H}^{G_l}_{P_l}},
\end{align}
where $P = \{P_1, \dots, P_L\}$ is the collection of probability distributions for all constraints. Correspondingly, the warm-started initial state is defined as:
\begin{align}
   \ket{\psi_0(P)} = \bigotimes_{l=1}^L \ket{W_{P_l}}.
\end{align}

The algorithm initializes in an equal superposition, meaning $P^{(0)}_{l,i} = 1/k_l$ for all variables $i$ within constraint $l$. In the first step, we optimize the variational parameters $\vec{\beta}$ and $\vec{\gamma}$ for a $p$-layer QAOA circuit. Because the warm-started initial state is explicitly constructed to be the ground state of the warm-started mixer, we assume the optimal parameter landscape does not shift drastically during the iterative probability updates~\cite{lopez-ruiz2025}. A numerical experiment reported in Sec.~\ref{sec:eval} will support this assumption. Consequently, we perform this parameter optimization only once and reuse the resulting $\vec{\beta}$ and $\vec{\gamma}$ throughout all subsequent steps of the algorithm. An empirical analysis on how much re-optimizing parameters will improve performance can be found in Appendix~\ref{app:reopt}.

\begin{algorithm}[t]
\caption{Warm-Start QAOA (WS-QAOA)}\label{alg:ws-qaoa}
\KwData{$H_C, U_M, P, M, \vec\beta, \vec\gamma$}
\KwResult{Sequence of measurement shots $X = (x_1, \dots, x_M)$}
\For{$s = 1$ \KwTo $M$}{
    $\ket{\psi} \gets \bigotimes_{l=1}^L \ket{W_{P_l}}$ \Comment*{initialize state} 
    
    \For{$i = 0$ \KwTo $p-1$}{
        $\ket{\psi} \gets e^{-i\gamma_i H_C} \ket{\psi}$ \Comment*{apply cost operator} 
        $\ket{\psi} \gets U_M(P, \beta_i) \ket{\psi}$ \Comment*{apply WS-mixer} 
    }
    Measure $\ket{\psi}$ to obtain bitstring $x_s$\;
}
\end{algorithm}

To obtain an initial biased distribution $P^{(1)}$, we execute WS-QAOA (Algorithm~\ref{alg:ws-qaoa}) using the uniform distribution $P^{(0)}$. We measure $M$ shots, yielding bitstrings $X^{(0)} = \{x^{(0)}_1, \dots, x^{(0)}_M\} $ and their corresponding objective energies $E_m^{(0)} = C(x_m^{(0)})$. We then compute the updated probabilities using a Boltzmann-weighted expectation value over the measured samples, similar to~\cite{lopez-ruiz2025}:
\begin{align}\label{eq:prob-update}
    P_{l,i}^{(t+1)} = \frac{1}{Z_l} \sum_{x\in X^{(t)}} e^{-\beta_T C(x) / \Delta^{(t)}} x_{l,i},
\end{align}
where $Z_l$ is chosen to normalize the probabilities for each constraint $l$, $\beta_T$ is the inverse temperature, and $\Delta^{(t)} = \max_m E_m^{(t)} - \min_m E_m^{(t)}$ represents the energy spread of the samples at iteration $t$. 

Note that the symbol $\beta$ appears in two distinct contexts in this framework: as the variational QAOA parameters ($\vec{\beta}$) and as the inverse temperature ($\beta_T$). Subscript $T$ will always refer to the inverse temperature.

\begin{algorithm}[t]
\caption{Iterative Warm-Start QAOA (IWS-QAOA)}\label{alg:iws-qaoa}
\KwData{$H_C, U_M, \{k_1, \dots, k_L\}, p, M, \overline{M}, \beta_T, \epsilon$}
$P_{l,i}^{(0)} \gets 1 / k_l \quad \forall l = 1, \dots, L, \;\forall i = 1, \dots, k_l$\;
$M_\text{agg} \gets 0$\;
$t \gets 0$\;
$\vec\beta, \vec\gamma \gets \mathrm{OptimizeQAOA}(H_C, U_M, P^{(0)}, p)$\;
\While{$M_{\text{agg}} < \bar{M}$}{
    $X^{(t)} \gets \mathrm{WS\text{-}QAOA}(H_C, U_M, P^{(t)}, M, \vec\beta, \vec\gamma)$\;
    $E_m^{(t)} \gets C(x_{m}^{(t)}) \quad \forall m \in \{1, \dots, M\}$\;
    $\Delta^{(t)} \gets \max_m E_m^{(t)} - \min_m E_m^{(t)}$\;
    $P^{(t+1)} \gets \mathrm{UpdateProbabilities}(X^{(t)}, E^{(t)}, \Delta^{(t)}, \beta_T$ \Comment*{based on Eq.~\eqref{eq:prob-update}}
    $P^{(t+1)} \gets \mathrm{Clamp}(P^{(t+1)}, \epsilon)$\;
    $M_\text{agg} \gets M_\text{agg} + M$\;
    $t \gets t + 1$\;
}
\end{algorithm}

To prevent the distribution $P_{l,i}^{(t)}$ from becoming overly concentrated on a single state $i$, resulting in detrimental QAOA performance~\cite{cain2023}, we clamp the probabilities to the interval $[\frac{\epsilon}{k_l-1}, 1-\epsilon]$, similar to the regularization approaches in Refs.~\cite{egger2021, carmo2025}. Using these updated probabilities, we run WS-QAOA again to generate new samples, iterating until we have accumulated a total of $\bar{M}$ shots. The complete procedure is formalized in Algorithm~\ref{alg:iws-qaoa}.

The Iterative Warm-Start QAOA (IWS-QAOA) relies on the three essential hyperparameters $\epsilon$, $ \beta_T$, and $M$, whose influences are outlined below:
\begin{itemize}
    \item \textbf{Regularization ($\epsilon \in [0, 1-1/k_l]$):} This parameter determines the strength of algorithmic exploitation. If $\epsilon$ is too small, the algorithm may prematurely converge to a local minimum.
    \item \textbf{Inverse temperature ($\beta_T$):} This governs how aggressively the update procedure biases the search toward lower-energy solutions. A larger $\beta_T$ accelerates convergence but also increases the risk of the algorithm becoming trapped in local minima.
    \item \textbf{Shots per iteration ($M$):} This defines the sample size drawn in each step. A smaller $M$ reduces the quantum resources required per iteration and accelerates the update cycle, but it introduces higher statistical noise that can lead to local minima. However, a small $M$ also presents a distinct mathematical advantage: if the objective function possesses internal symmetries, a smaller sample size facilitates symmetry breaking, naturally biasing the algorithm toward a specific optimal sector. For instance, Max-Cut has an internal $Z_2$ symmetry~\cite{tsvelikhovskiy2026}: for every solution exists a solution with flipped bits but equal energy. Evaluating Eq.~\ref{eq:prob-update} exactly would result in non-biased probabilities, whereas shot-based evaluation helps push the probabilities to one sector.
\end{itemize}

As a stochastic hybrid quantum heuristic, IWS-QAOA accelerates the search for optimal solutions while potentially demanding fewer quantum resources than default QAOA. Note that only $p$ (QAOA layers) and $\overline{M}$ (total shots) are responsible for the quantum resource requirements of IWS-QAOA. The re-evaluation of probabilities comes at virtually no cost for a classical computer. Due to the inherent risk of converging to local minima, it is highly recommended to execute multiple independent repetitions of the algorithm. 

Sec.~\ref{sec:eval} will demonstrate the practical efficacy of IWS-QAOA through rigorous numerical experiments. Notably, because the Boltzmann-weighted expectation value serves as a purely classical selection mechanism, the entire algorithmic loop can be executed without QAOA by substituting the quantum circuit with classical random sampling from the iteratively updated probability distributions. This serves as a natural baseline in our benchmarks.
\section{Numerical Simulations}\label{sec:eval}

\begin{figure*}[t]
    \centering
    \includegraphics[width=\linewidth]{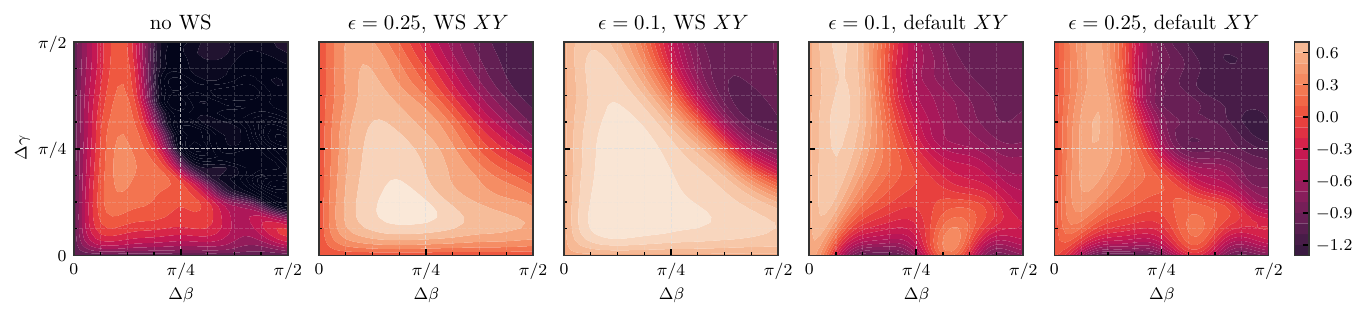}
    \caption{Contour plots showing the energy landscape in terms of approximation ratio of WS-QAOA for $p=5$ on a TSP instance with $N=7$. The leftmost plot applies no warm starting; the following two plots show the landscape with warm starting based on the ideal solution, using regularization with $\epsilon = 0.25$ and $\epsilon = 0.1$. The final two plots show warm-starting using $\ket{W_p}$ and the default XY-mixer, which is not aligned with the initial state.\label{fig:landscape}}
\end{figure*}

\subsection{Problem Instances}

\subsubsection{Max-\textit{k}-Cut}
QAOA was initially developed as an approximate algorithm for Max-Cut~\cite{farhi2014a}. Max-Cut remains a central optimization problem for benchmarking QAOA and variants like warm-starting~\cite{yu2025, egger2021}. Max-$k$-Cut (MkC) is the natural extension that separates the nodes into $k$ partitions rather than two. For a graph $G(V, E)$ with $|V| = N$ and edge weights $w_{uv}$, it is formulated using one-hot encoding as:
\begin{align}
    \begin{split}
    \max_x& \sum_{u,v \in E} w_{uv} \left(1 - \sum_{i = 1}^k x_{u,i} x_{v,i}\right) \\ \text{s.t.}& \,\, \sum_{i=1}^k x_{v,i} = 1\,\,\forall v\in V.
    \end{split}
\end{align}
Since the constant term can be omitted, we implement the equivalent minimization problem:
\begin{align*}
     \min_{x'}& \sum_{u,v} w_{uv} \sum_{i = 1}^k x_{u,i} x_{v,i}   \quad \text{s.t.} \,\, \sum_{i=1}^k x_{v,i} = 1\,\,\forall v\in V\setminus\{1\},
\end{align*}
where we additionally fix the first node to be in category 1, i.e., $x_{1,1} = 1$ and $x_{1,i} = 0$ for all $i>1$. Here, $x'$ denotes all non-fixed variables. Fixing this node breaks the natural $S_k$ symmetry of the problem, which assists the IWS-QAOA biasing.

\paragraph*{Problem Instances} Similar to the Max-Cut instances from Ref.~\cite{egger2021}, we use a complete graph $K_N$ and sample $w_{uv}$ from a uniform distribution over the set $\{-1, -0.9, \dots, 1\}$. We consider problem scenarios with $N \in \{12, 16\}$ for $k=3$, and $N \in \{10, 14\}$ for $k=4$. For each scenario, we generate five distinct instances.

\subsubsection{Traveling Salesperson Problem}
One of the most famous combinatorial optimization problems is the Traveling Salesperson Problem (TSP), which seeks to find the shortest cycle connecting all nodes in a given fully connected graph $K_N$ with edge weights $w_{uv} > 0$. It is naturally formulated as a quadratic integer program:
\begin{align}
\begin{split}
    \min_x& \sum_{u,v\in E}w_{uv} \sum_{t=1}^N x_{u,t}x_{v,(t+1)\%N} \\ \text{s.t.}& \sum_v x_{v,t} = 1 \,\,\forall t, \text{ and } \sum_t x_{v,t} = 1\,\, \forall v\in V.
    \end{split}
\end{align}

In contrast to MkC, the TSP formulation features two overlapping sets of one-hot constraints, which effectively define a permutation matrix. Because XY-mixers can only enforce non-overlapping one-hot constraints, we must encode one of these constraint sets as a quadratic penalty in the objective function. Furthermore, to break symmetry, we fix the starting city to be visited first (i.e., $x_{1,1} = 1$ and $x_{1,t} = 0$ for all $t>1$). This yields the revised formulation:
\begin{align}
\min_{x'}& \sum_{u,v} w_{uv} \sum_{t=1}^N x_{u,t}x_{v,t+1} + \lambda \sum_{t=2}^N \left(\sum_{v\in V\setminus\{1\}} x_{v,t} - 1\right)^2\nonumber \\
\text{s.t.}& \sum_{t=2}^N x_{v,t} = 1 \,\;\forall v\in V\setminus\{1\},
\end{align}
where $x'$ denotes the non-fixed variables.

\paragraph*{Problem Instances} For each problem size $N \in \{6, 7, 8, 9\}$, we generate five distinct instances by placing cities equally spaced around a circle of radius two. Each city's radius is then offset by sampling from a normal distribution with a standard deviation of 1, inspired by Refs.~\cite{schawe2016, bucher2024}. The penalty parameter is fixed at $\lambda = 2$.

\subsection{Metrics, Parameters, and Simulation Technique}

The approximation ratio is defined through the expectation value $E = \bra{\psi}H_C\ket{\psi}$ and the optimal solution energy $E_\text{opt}$:
\begin{align}\label{eq:approx_ratio}
    r = 1 - \frac{|E - E_\text{opt}|}{|E_\text{opt}|},
\end{align}
serving as a way to quantify the quality of the solution ensemble sampled from the QAOA.

With IWS-QAOA, we are particularly interested in how fast the algorithm finds a high-quality solution. To this end, we compute the expected running best solution energy trace as follows:
\begin{align}\label{eq:exp_trace}
    E_s = \sum_x |\braket{x|\psi}|^2 \begin{cases}C(x) &\text{if } C(x) < E_{s-1} \\ E_{s-1} &\text{else} \end{cases},
\end{align}
where $E_0 = E$. The metric $E_s$ describes the expected best solution after $s$ shots drawn from the algorithm, and it is expected to converge toward the optimal solution if the probability of sampling it is non-vanishing. Note that $E_s$ is not the exact expected running best solution; instead, it is a recursively evaluated approximation of that quantity, giving us a connected and monotonic curve throughout IWS restarting. Normalizing $E_s$ according to Eq.~\eqref{eq:approx_ratio} yields the approximation trace $r_s$.

Finally, the probability $P_\text{opt}$ of obtaining an optimal solution, defined as
\begin{align}
    P_\text{opt} = \sum_{x \in \{x | C(x) = E_\text{opt}\}} |\braket{x|\psi}|^2,
\end{align}
is also a vital metric for benchmarking quantum optimization algorithms~\cite{bucher2024}.

For the hyperparameters of IWS-QAOA, we chose a regularization of $\epsilon = 0.2$ and an inverse temperature of $\beta_T = 15$, using varying numbers of shots $M\in \{100, 200, 500\}$. These parameters were identified during a preliminary hyperparameter study. We found that the algorithm's performance is relatively robust, meaning that slight variations ($\epsilon \pm 0.1$, $\beta_T \pm 10$) yield comparable results. Furthermore, we use a complete graph for the mixer topology. For the QAOA variational parameter optimization, we rely on the BFGS algorithm~\cite{fletcher2008, 2020SciPy-NMeth} to optimize a linear schedule consisting of four parameters $\{\beta_0, \Delta\beta, \gamma_0, \Delta\gamma\}$:
\begin{align}\label{eq:linear-schedule}
    \beta_i = \beta_0 - \frac{i \Delta \beta }{p-1}, \quad \gamma_i = \gamma_0 + \frac{i \Delta \gamma}{ p-1}.
\end{align}
For $p=1$, we optimize only $\beta_0$ and $\gamma_0$ and omit the rest. Due to the algorithm's heuristic nature, we repeat IWS-QAOA ten times for each problem instance.

Lastly, for simulating the QAOA circuits, we exploit the fact that each qubit set associated with a one-hot constraint never leaves the Hamming-weight 1 subspace. Therefore, we only need to track $k_l$ statevector entries per constraint instead of $2^{k_l}$. In total, the effective state vector size is reduced to $\prod_l k_l$, building on methods from Ref.~\cite{bucher2025}. We also precompute the diagonal of the cost Hamiltonian by evaluating the objective over all feasible states, following prior QAOA simulation work~\cite{stein2024, lykov2023, golden2023}. This simulation technique allows us to simulate up to an $N=9$ TSP instance ($8 \times 8 = 64$ qubits, but with a drastically reduced state vector size of $8^8 = 2^{24}$) on a consumer-grade GPU.

\subsection{Experimental Validation of Warm-Started XY-Mixer}\label{sec:parameter-validation}

Before analyzing the performance of the hybrid algorithm, we verify the warm-started XY-mixer developed in Sec.~\ref{sec:ws-mixers} by directly comparing its energy landscapes with those of the non-warm-started, default XY-mixer. We choose a TSP instance with $N=7$ and use the ideal solution as the warm-start state with different regularization strengths $\epsilon$ to isolate the difference, and compare the default XY-mixer to our aligned, warm-started XY-mixer.

Fig.~\ref{fig:landscape} shows the energy landscape in terms of the approximation ratio for a grid of linear QAOA parameter schedules using $\Delta\beta$ and $\Delta\gamma$, with $\beta_0 = \Delta\beta (p - 1/2) /p$ and $\gamma_0 = \Delta\gamma/(2p)$. We observe that warm-starting improves the approximation ratio regardless of whether the XY-mixer is aligned (the brighter areas become larger). Furthermore, as expected, smaller values of $\epsilon$ yield better solution qualities. Most importantly, fully warm-starting the XY-mixer profoundly affects the energy landscape. While the landscapes of the warm-started XY-mixer closely resemble the non-warm-started baseline, the optimization landscape significantly degrades when the default XY-mixer is used with a biased initial state. Additionally, the region of high-quality parameterizations expands for warm-started mixers while retaining the characteristic triangular shape identified in Ref.~\cite{montanez-barrera2025}. This confirms that aligning the XY-mixer with the initial state $\Wp$ is highly advantageous.

Lastly, we address the scaling of the warm-started XY-mixer block introduced in Sec.~\ref{sec:scaling}. Without this scaling (not shown), the landscape would stretch along the $\Delta \beta$ axis as $\epsilon$ decreases. With the scaling applied, the optimal $\Delta \beta$ values remain largely invariant. Consequently, high-quality parameters from the non-warm-started landscape remain effective in the warm-started scenario, justifying the optimize-once parameter strategy employed in IWS-QAOA. Further analysis on re-optimization is presented in Appendix~\ref{app:reopt}.

\begin{figure*}
    \centering
    \subfloat[Approximation ratio\label{fig:mkc_approx_ratio}]{\includegraphics[width=0.5\linewidth]{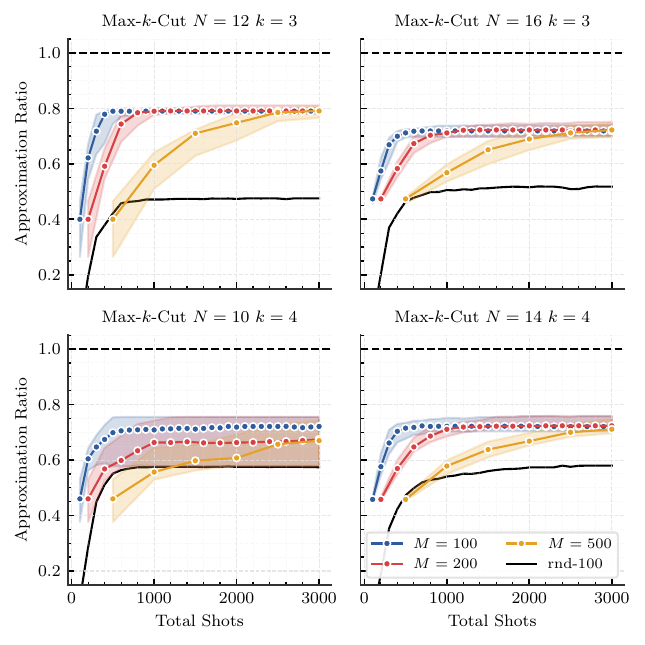}}%
    \subfloat[Approximation trace following Eq.~\eqref{eq:exp_trace}\label{fig:mkc_trace}]{\includegraphics[width=0.5\linewidth]{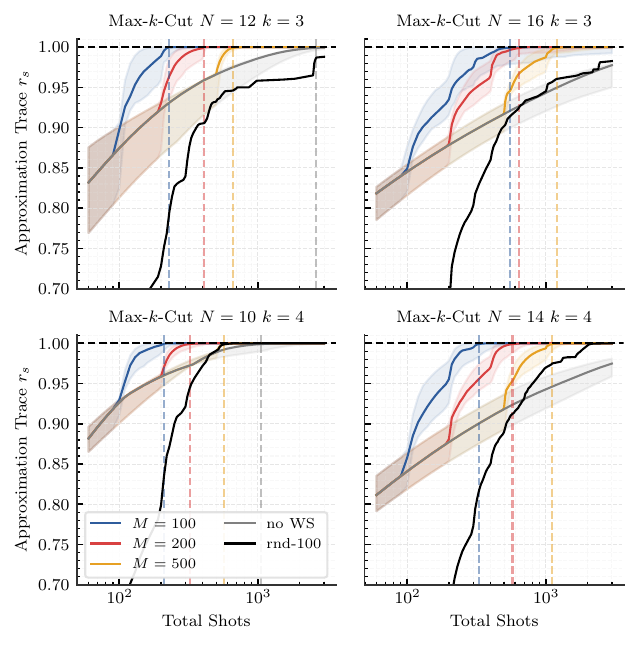}}
    \caption{Approximation ratio~(a) and approximation trace~(b) of IWS-QAOA at $p=1$ as a function of the total number of shots ($M_\text{agg}$ from Algorithm~\ref{alg:iws-qaoa}) for $M \in \{100, 200, 500\}$ and $\overline{M} = 3000$ across four MkC instance classes. Each panel displays the median over five instances, with ten runs per instance. Error bands indicate the interquartile range. The solid black line represents the median performance of IWS using classical random sampling with $M = 100$ (the best-performing sample size among those tested). Vertical dashed lines in panel~(b) mark the median number of total shots required to identify the optimal solution.}
    \label{fig:mkc_results}
\end{figure*}

\begin{figure}
    \centering
    \includegraphics[width=\linewidth]{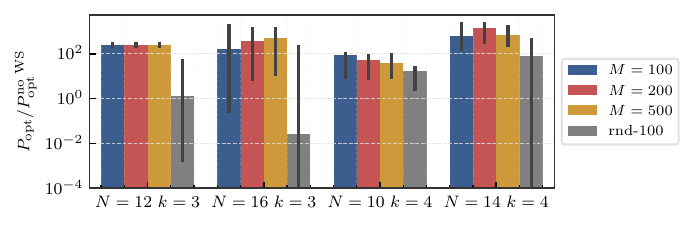}
    \caption{Median improvement ratio of the optimal solution probability, $P_\text{opt}$, achieved by IWS-QAOA compared to the baseline without warm-starting (no WS) across four MkC problem classes for various sample sizes $M$. Error bars indicate the interquartile range. The gray bars represent the performance of IWS using classical random sampling with $M=100$. For IWS-QAOA, $P_\text{opt}$ is evaluated directly from the state vector following the final iteration of IWS-QAOA.}
    \label{fig:mkc_popt}
\end{figure}

\begin{figure*}
    \centering
    \includegraphics[width=\linewidth]{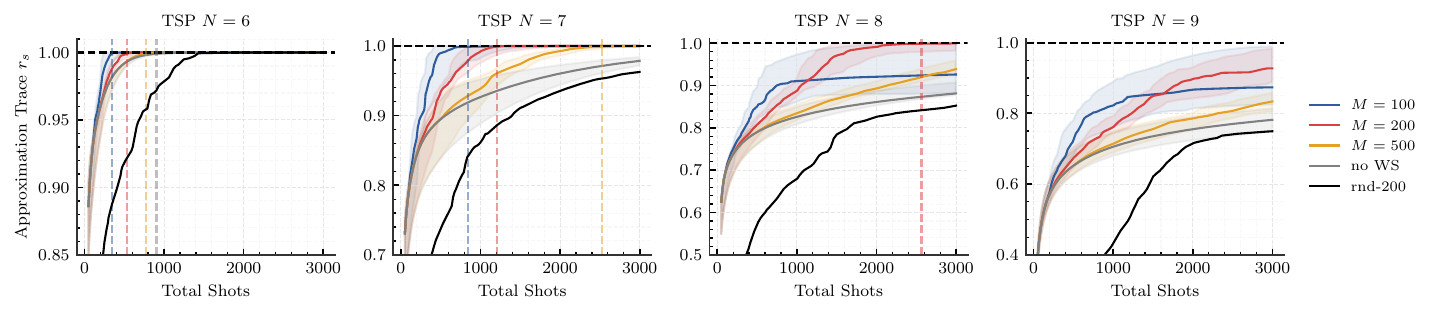}
    \caption{Approximation trace of IWS-QAOA at $p=1$ as a function of the total number of shots for $M\in \{100, 200, 500\}$ and $\overline{M} = 3000$ on TSP instances ranging from 6 to 9 cities. Each panel displays the median over five instances, with ten runs per instance. Error bands indicate the interquartile range. The solid black line represents the median baseline performance of IWS using classical random sampling with $M=200$, which outperformed $M=100$ for these TSP instances.} 
    \label{fig:tsp_trace}
\end{figure*}

\begin{figure}
    \centering
    \includegraphics[width=\linewidth]{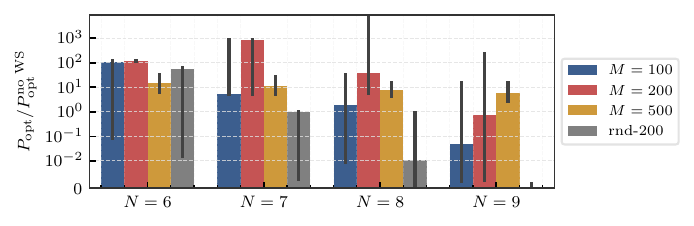}
    \caption{Median improvement ratio of $P_\text{opt}$ achieved by IWS-QAOA relative to the baseline without warm-starting (WS) for the four TSP instance sizes across various $M$ values. Error bars indicate the interquartile range. The gray bars represent IWS with classical random sampling at $M=200$. For IWS-QAOA, $P_\text{opt}$ is evaluated directly from the state vector following the final iteration of IWS-QAOA.}
    \label{fig:tsp_popt}
\end{figure}

\begin{figure*}
    \centering
    \includegraphics[width=\linewidth]{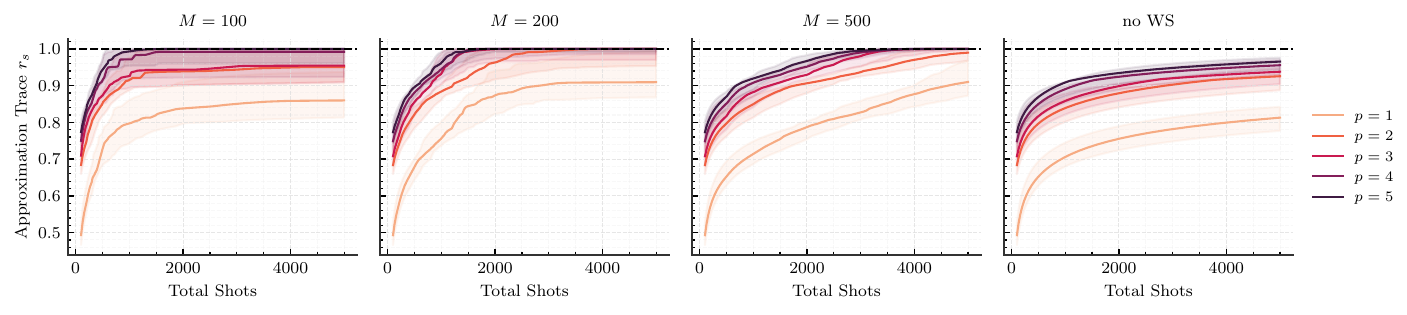}
    \caption{Approximation trace of IWS-QAOA for circuit depths $p\in \{1,\dots,5\}$ as a function of the total number of shots. Results are shown for $M\in \{200, 500\}$ and $\overline{M} = 5000$ on the 9-city TSP instances. Each panel displays the median over five instances, with ten runs per instance, and error bands indicating the interquartile range.}
    \label{fig:tsp_reps_trace}
\end{figure*}

\begin{figure}
    \centering
    \includegraphics[width=\linewidth]{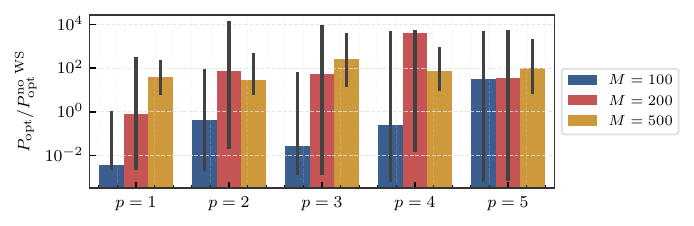}
    \caption{Median improvement ratio of $P_\text{opt}$ achieved by IWS-QAOA relative to the non-warm-started baseline for the 9-city TSP instances across circuit depths $p\in\{1,\dots,5\}$. Error bars indicate the interquartile range. For IWS-QAOA, $P_\text{opt}$ is evaluated directly from the state vector following the final iteration.}
    \label{fig:tsp_reps_popt}
\end{figure}

\subsection{Simulation Results}
We now analyze the IWS-QAOA simulation results for the MkC and TSP problems.

\subsubsection{Max-\textit{k}-Cut}

Fig.~\ref{fig:mkc_approx_ratio} shows the approximation ratio of $p=1$ IWS-QAOA with respect to the total shots gathered throughout the algorithm execution. It is apparent that IWS-QAOA---independent of $M$---improves upon the base QAOA approximation ratio (which corresponds to the performance at the first data point). Furthermore, we observe that all values of $M$ converge to the same saturation level. Due to the regularization factor $\epsilon$, exact convergence to $1$ is precluded. In the $N=10, k=4$ case, the convergence levels vary slightly, and the variance of the values increases. Expectedly, the median of the random runs (black) also improves drastically compared to the non-warm-started case. Yet, the saturation level is significantly worse than that of IWS-QAOA, except for the $N=10, k=4$ case. This data also suggests that smaller $M$ values converge faster, as they undergo more warm-starting update iterations for the same number of total shots.

Fig.~\ref{fig:mkc_trace} shows the expected best solution at the current shot drawn from the algorithm, which effectively acts as a solution quality versus runtime plot. For the baseline no-warm-start QAOA, we observe a smooth curve that slowly approaches 1, suggesting that we will eventually sample the optimal solution. For $(N,k) = (12,3)$ and $(10, 4)$, this occurs at approximately 3000 and 1000 shots, respectively. IWS-QAOA finds the optimal solution in all cases, with smaller values of $M$ converging more rapidly. For $(10, 4)$, only about 200 shots suffice, while for $(16, 3)$ approximately 550 shots are required at $M = 100$. Finally, the random sampling baseline demonstrates that for $k=4$, the random IWS algorithm successfully finds the optimal solution, and does so relatively quickly ($<1000$ shots). However, at $k=3$, this is not the case. From these results, we can deduce that $(10, 4)$ is a relatively easy instance, also supported by the baseline QAOA solution quality being comparatively good. Consequently, the random sampling heuristic also performs well on this instance.

Finally, Fig.~\ref{fig:mkc_popt} compares the $P_\text{opt}$ of the baseline QAOA against the IWS-QAOA $P_\text{opt}$ after the final iteration. We observe a roughly two-orders-of-magnitude increase in probability across all instances, independent of the selected $M$. Additionally, at $k = 4$ (which appear to be the easier instances), random IWS also improves over the baseline QAOA. Conversely, at $k=3$, there is no clear improvement for the random algorithm; instead, the median is $\leq 1$, and the variance is high.

We conclude that IWS-QAOA successfully accelerates the solution-finding process for MkC instances. Although using $M=100$ is more prone to statistical noise and becoming trapped in local minima, it proves sufficient for the specific MkC instances investigated here. Still, we observed that random sampling with IWS is a viable optimization method for the MkC instances in question, especially for $k=4$. 

\subsubsection{Traveling Salesperson Problem}

Fig.~\ref{fig:tsp_trace} displays the approximation trace of the $p=1$ IWS-QAOA across different TSP instance sizes. In the smallest case ($N=6$), all methods find the optimal solution in fewer than 2000 shots. However, for $N \geq 7$, some methods begin to fail to identify the optimal solution. The non-warm-started QAOA baseline does not reach an approximation trace of 1 within 3000 shots. Similarly, the random IWS baseline fails to find the optimal solution in these cases and remains below the non-warm-started QAOA. For $N\geq 8$, we also observe that choosing a sample size $M < 200$ that is too small can trap the algorithm in local minima. Although $M=100$ initially increases most steeply, its rate of improvement slows drastically after approximately 800 shots. While $M=500$ steadily increases solution quality, the more runtime-intensive iterations hinder it from surpassing $M=200$ in performance. Furthermore, its large associated interquartile range indicates that IWS-QAOA converges to various local minima at differing distances from the optimal solution. Overall, however, IWS-QAOA finds better solutions more quickly than the standard QAOA.

We hypothesize that convergence to local minima drastically reduces the $P_\text{opt}$ ratio, thereby increasing the probability of sampling sub-optimal but high-quality solutions. Fig.~\ref{fig:tsp_popt} confirms this assumption. For the smallest instances, IWS-QAOA improves the optimal sampling probability across all values of $M$. However, as the problem size increases, the improvement ratio drops, particularly for small $M$, which aligns with our local minimum hypothesis. Still, IWS-QAOA at $M=500$ and $N=9$ improves $P_\text{opt}$ by roughly a factor of ten, even though its approximation trace in Fig.~\ref{fig:tsp_trace} barely deviates from the baseline.

Our analysis thus far has focused on the $p=1$ case, which, while sufficient for the previously evaluated (small) MkC instances, falls short for the TSP. Therefore, we run additional experiments at circuit depths of $p\in\{1,\dots,5\}$ for the 9-city TSP instances. Fig.~\ref{fig:tsp_reps_trace} depicts the approximation traces for $M\in\{100, 200, 500\}$. Without warm-starting, increasing the circuit depth lifts the trace toward the optimum, yet even at $p=5$ standard QAOA remains insufficient to reliably obtain the optimal solution. Conversely, IWS-QAOA consistently finds the optimal solution for $M=200$ and $M=500$, provided $p\geq3$, but only for $p=5$ with $M=100$. Again, a smaller $M$ leads to faster convergence, but with the trade-off of not finding the ideal solution.

Fig.~\ref{fig:tsp_reps_popt} shows the resulting improvement ratio of $P_\text{opt}$. Here, it is clearly visible that $M=200$ benefits significantly from better initial solutions. While the median of runs at $p=1$ converged to local minima, circuit depths of $p\geq2$ improve the baseline probability by one to two orders of magnitude. For $M=100$, we observe sufficient improvement only at $p=5$. However, these values exhibit a significantly larger variance compared to $M=500$, which consistently improves the baseline $P_\text{opt}$ (by $\approx10\times$ for $p\leq2$ and $>100\times$ for $p \geq 3$).

Overall, these results demonstrate that IWS significantly reduces the runtime required to sample high-quality or optimal solutions when the underlying QAOA circuit produces sufficient-quality samples. Conversely, if the baseline QAOA inherently struggles to find near-optimal solutions, the IWS routine is highly prone to converging to local minima.

\section{NISQ-Hardware Evaluation}\label{sec:hw}

Choosing a hardware-friendly XY-mixer topology facilitates low-depth circuit transpilation of the mixer to the hardware graph and native gate set. However, the problem-specific cost Hamiltonian, $H_C$, must also be mapped to the hardware topology. This typically requires integrating \textsc{swap} gates to embed the problem's interaction graph into the physical qubit connectivity. While promising ansätze exist for efficient embedding~\cite{weidenfeller2022, matsuo2023, qaio25}, circuit depths are still expected to surpass the threshold for extracting meaningful results, especially as the qubit count scales (assuming a maximum viable two-qubit-gate depth of approximately 100 for current IBM devices). To circumvent this issue, we do not transpile standard MkC or TSP Hamiltonians to the hardware. Instead, we generate a constrained spin-glass-like problem with hardware-tailored instances that inherently require shallow implementation depths, drawing inspiration from prior work~\cite{pelofske2023, chandarana2025, kotil2025}. Despite being highly specific, these sparse instances allow the execution of problem sizes that fill the entire device, whereas dense problem instances would lead to excessively deep circuits. For our benchmarks, we generated five 144-qubit problem instances to evaluate IWS-QAOA on the \texttt{ibm\_boston} Heron r3 QPU. Because NISQ hardware is inherently noisy, the sampled bitstrings frequently fail to satisfy all one-hot constraints of the input problem. To counteract this, we employ a classical post-processing step to correct any constraint violations.

\subsection{Problem Instances}

\begin{figure*}
\subfloat[Coupling map \texttt{ibm\_boston}\label{fig:coupling_map}]{\includegraphics[width=0.45\textwidth]{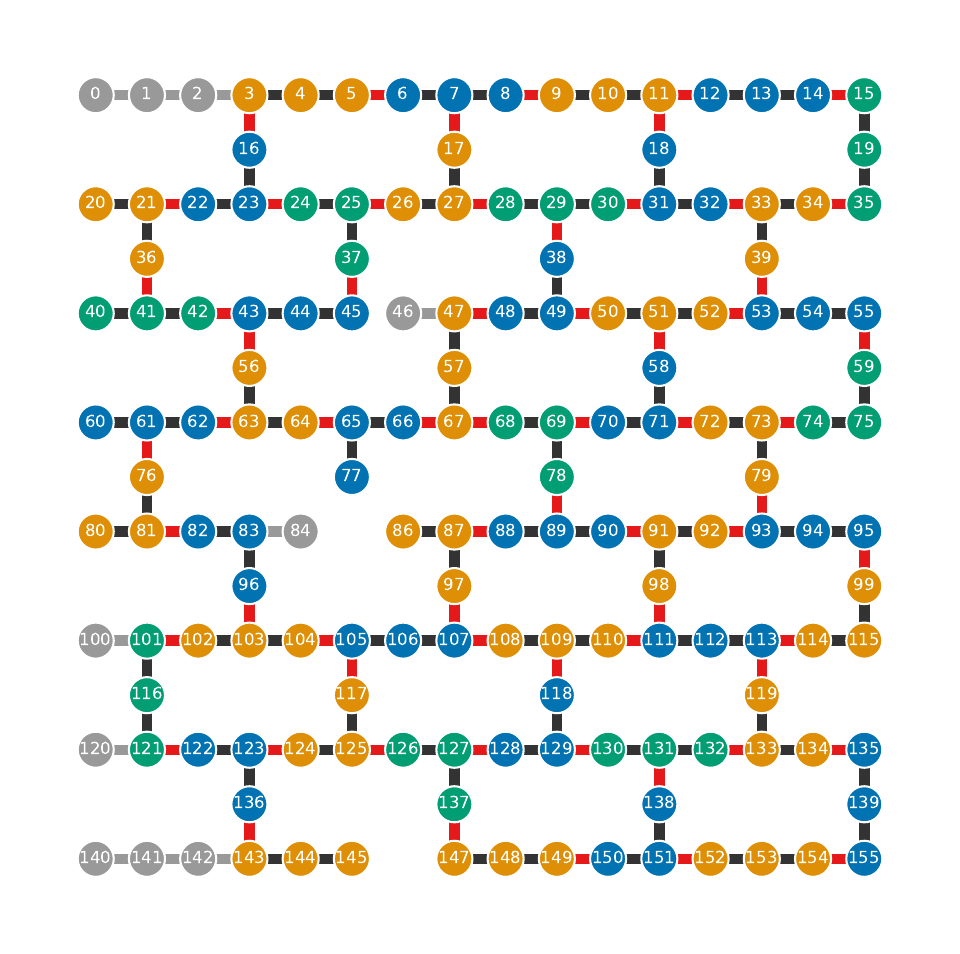}}
\hfil
\subfloat[Interaction graph \texttt{ibm\_boston}\label{fig:interaction_graph}]{\includegraphics[width=0.45\textwidth]{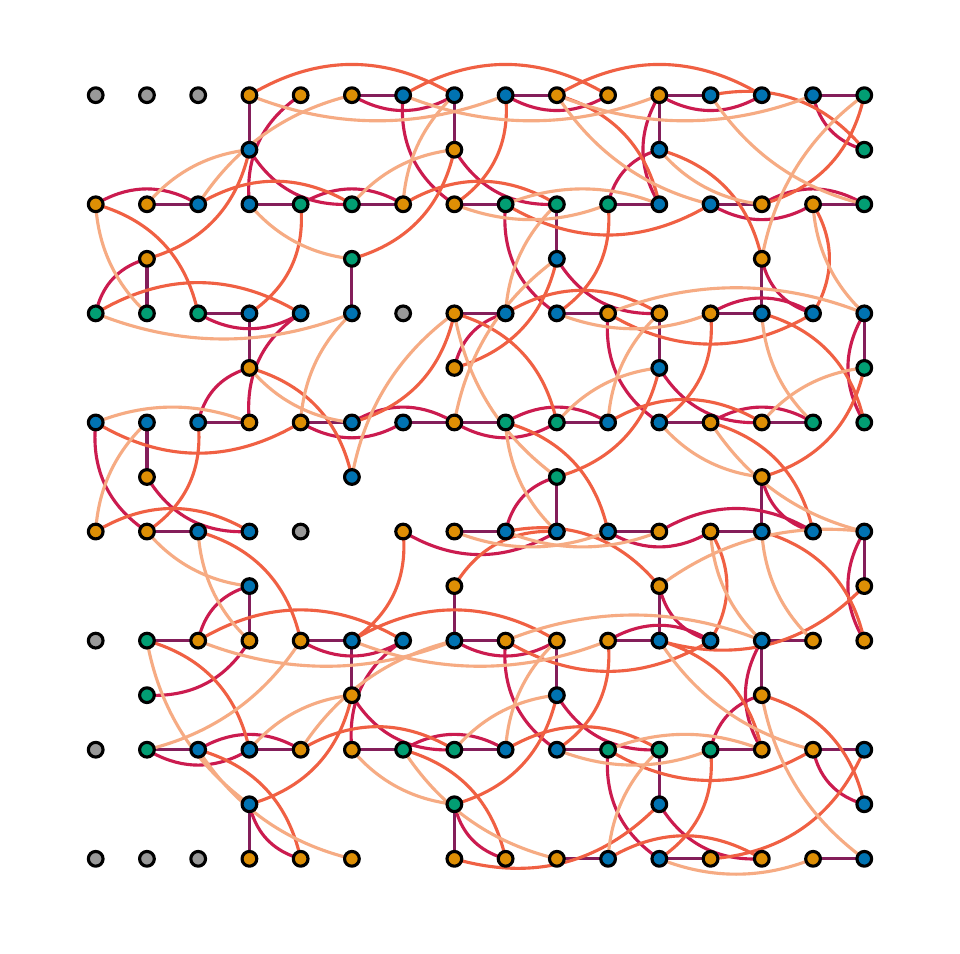}}
\caption{Hardware-tailored 144-qubit problem instance for the \texttt{ibm\_boston} QPU, encompassing 48 one-hot constraints that each span three binary variables. (a) The device coupling map highlighting the placed qubit triplets. Qubits belonging to the same triplet share a distinct color, while inactive qubits and couplers are depicted in gray. Black edges denote couplers reserved for \textsc{swap} gates or the warm-started XY-mixer implementation. Red edges represent the physical couplers used to embed the optimization problem's interactions. (b) The effective interaction graph of the resulting problem instance. Red edges indicate two-body terms used to model the problem, each assigned a random weight drawn uniformly from $\{-1, -0.9, \dots, 1\}$. The color shade indicates the \textsc{swap} layer during which the coupling is realized. In total, the instance supports 241 distinct interactions.}
\end{figure*}

\begin{table}[]
    \centering
    \setlength{\tabcolsep}{6pt}
\begin{tabular}{lrrrrrr}
\toprule
 $p$ & depth & depth-2Q & $\sqrt{X}$ & $R_Z$ & \textsc{cz} & $X$ \\
 \midrule
1 & 116 & 32 & 2937 & 2205 & 1232 & 99 \\
2 & 213 & 59 & 5431 & 3893 & 2312 & 200 \\
3 & 309 & 86 & 7945 & 5590 & 3400 & 295 \\
\bottomrule
\end{tabular}
    \caption{Transpiled circuit metrics and native gate counts of the hardware-tailored instances for varying QAOA depths ($p$).}
    \label{tab:circuit-metrics}
\end{table}

To generate problem instances tailored to the heavy-hex topology of an IBM QPU, we first filter the device's coupling map (see Fig.~\ref{fig:coupling_map}) by removing any qubits or couplers with unacceptably high error rates (specifically, \textsc{CZ} errors $>5\%$ and readout errors $>30\%$). Next, we identify adjacent qubit triplets---which encode the binary variables for each one-hot constraint---such that the inter-triplet connectivity is maximized while the error rates of the selected couplers are minimized. The specific optimization problem solved to determine this optimal triplet placement is detailed in Appendix~\ref{sec:hw-optimization}. The final mapping, with 144 selected qubits encompassing $N=48$, 3-bit one-hot constraints, is illustrated in Fig.~\ref{fig:coupling_map}.

The red edges in Fig.~\ref{fig:coupling_map} represent the physical couplings available to model the optimization objective. Because relying solely on these initial interactions would yield an extremely sparse and overly simplistic optimization problem, we artificially increase the complexity by inserting three layers of \textsc{swap} gates. Within each triplet, these \textsc{swap} operations are applied sequentially between qubits $(1, 2)$, $(2, 3)$, and again $(1, 2)$. After each \textsc{swap} layer, a new set of physical interactions becomes accessible due to the altered adjacency of the encoded binary variables. The composite interaction graph across all four \textsc{swap} phases (before, between, and after the \textsc{swap} layers) is visualized in Fig.~\ref{fig:interaction_graph}, yielding a total of 241 realizable quadratic terms on \texttt{ibm\_boston}. The overall optimization problem can thus be formulated as:
\begin{align}
    \min_x \sum_{(i,j) \in I} w_{ij} x_{i} x_{j} \quad\text{s.t.}\quad \sum_{i\in t} x_i = 1 \quad \forall t \in \mathcal{T}^*,
\end{align}
where $\mathcal{T}^*$ represents the set of all selected qubit triplets and $I$ denotes the set of accessible interconnections shown in Fig.~\ref{fig:interaction_graph}. To generate five distinct problem instances, the weights $w_{ij}$ are uniformly sampled from $\{-1, -0.9, \dots, 1\}$.

A key advantage of our chosen \textsc{swap}-layer configuration, compared to more generic routing strategies~\cite{chandarana2025}, is that the qubit triplets remain spatially confined; their internal order is simply reversed. This allows us to apply a linear XY-mixer topology within each QAOA layer without uncomputing the \textsc{swap} operations. Furthermore, to implement the subsequent QAOA layer, we can simply execute the \textsc{swap} and interaction layers in reverse.

Concerning topology, we implement the XY-line-mixer for each qubit-triplet, defined by the topology $G(V,E)$ with vertices $V=\{1, 2, 3\}$ and edges $E = \{(1, 2), (2, 3)\}$, being connected but not regular. Consequently, following Corollary~\ref{cor:app}, $\ket{W_P}$ remains the unique ground state of $\mathcal{H}^{G}_P$. The decomposition is given by
\begin{align*}
    e^{-i\beta \mathcal{H}_P^G} \approx e^{i\frac{\beta}{2} \ket{e_1}\bra{e_1}} e^{-i\frac{\beta}{2}\mathcal{H}_{2,3}(q_{2,3})}e^{i\frac{\beta}{2} \ket{e_3}\bra{e_3}} e^{-i\frac{\beta}{2}\mathcal{H}_{1,2}(q_{1,2})}.
\end{align*}

For a QAOA depth of $p=1$, the WS-QAOA circuits for these hardware-tailored instances transpile to a total circuit depth of 116 each, which corresponds to a two-qubit gate depth of approximately 32. Comprehensive circuit metrics and native gate counts for depths $p=1, 2,$ and $3$ are summarized in Table~\ref{tab:circuit-metrics}.

\begin{figure}
    \centering
    \includegraphics[width=\linewidth]{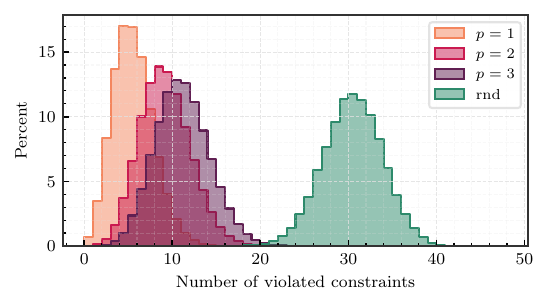}
    \caption{Histograms depicting the number of violated constraints per bitstring sampled from the WS-QAOA circuits. Each distribution aggregates 250\,000 measurement shots (comprising 5 problem instances, evaluated with 10 independent repetitions of 5000 shots each). The green histogram represents a random-sampling baseline exhibiting a mean violation rate of 67\%.}
    \label{fig:hist-constrs}
\end{figure}

\subsection{Post-Processing}\label{sec:postprocessing}

Although evolution under $U_M$ theoretically confines the system to the Hamming-weight-1 subspace for each constraint, hardware noise inevitably causes some measured samples to violate these constraints. We can model the occurrence of a violation within a single constraint by a failure probability $f$. Consequently, the success rate of sampling a completely feasible solution is $(1-f)^{48}$. To achieve a 99\% feasibility rate across the sampled bitstrings, the single-constraint failure rate would need to remain below $2\times10^{-4}$. Given that current IBM QPUs exhibit two-qubit gate errors on the order of $\sim10^{-3}$, attaining such a high feasibility ratio is unrealistic.

Fig.~\ref{fig:hist-constrs} illustrates the distribution of constraint violations across different QAOA depths. From this data, we deduce the empirical failure rates of single constraints to be $f(p=1) = 11\%$, $f(p=2) = 18\%$, and $f(p=3) = 22\%$. At $p=1$, we obtain only 1766 feasible samples, corresponding to a success rate of $0.7\%$ (leftmost bar). For $p=2$, this number drops to 34, and at $p=3$, only a single feasible solution is observed. The ratio of feasible solutions to possible binary assignments for this instance is given by $3^{48} / 2^{144} \approx 3.6 \times 10^{-21}$.

Consequently, discarding infeasible solutions via strict post-selection does not scale. Nevertheless, because the majority of constraint assignments within a given sample remain feasible (only a fraction of 10--20\% violating one or more constraints), we instead apply a classical post-processing (PP) step to repair infeasible assignments. Specifically, we employ a greedy steepest-descent algorithm based on a quadratically penalized QUBO cost function:
\begin{align}
    C(x) = \sum_{(i,j) \in I}w_{ij} x_{i} x_{j} + \lambda \sum_{t \in \mathcal{T}^*}\left(1-\sum_{i\in t} x_i \right)^2,
\end{align}
with the penalty parameter set to $\lambda =10$. The algorithm evaluates a candidate solution generated by the QAOA circuit by systematically flipping each bit, recording the corresponding change in the cost function, and then reverting the flip. If any bit flips yield a cost reduction, the algorithm permanently applies the single flip that produces the steepest descent. Each pass has run-time complexity $O(n)$, where $n$ is the number of binary variables. This process continues until a local minimum is reached and no further cost improvements are possible, corresponding to at most $n$ iterations and therefore a worst-case complexity of $O(n^2)$ per sample. Similar post-processing techniques have already been used in other hardware experiments~\cite{montanez-barrera2025, he2026}.

Notably, this greedy strategy inherently avoids altering bits associated with already-satisfied one-hot constraints. Because the algorithm only applies single-bit flips, transitioning between two valid one-hot assignments requires at least two flips. Consequently, any such transition would temporarily force the system into an infeasible state, which is heavily penalized by the objective function and thus rejected by the descent criteria.

Because this greedy algorithm serves as a classical heuristic in its own right, we also evaluate a baseline model that pairs random sampling with the same post-processing routine (rnd-PP) in the IWS-QAOA benchmark. Appendix~\ref{app:postprocessing} presents an in-depth study on the effects of the post-processing method.

\begin{table}
    \centering
    \setlength{\tabcolsep}{6pt}
    \begin{tabular}{lrrrr}
         \toprule
         $p$ & $\beta_0$ & $\Delta\beta$ & $\gamma_0$ & $\Delta \gamma$ \\
         \midrule
         1 & 0.90 & -- & 0.58 & -- \\
         2 & 1.20 & 0.96 & 0.41 & 0.58 \\
         3 & 1.36 & 1.13 & 0.32 & 0.63 \\
         \bottomrule 
    \end{tabular}
    \caption{Averaged optimal linear-schedule QAOA parameters, obtained from simulations of hardware-tailored problem instances of 24, 36, and 48 qubits.}
    \label{tab:qaoa-params}
\end{table}

\begin{figure*}
    \centering
    \includegraphics[width=\linewidth]{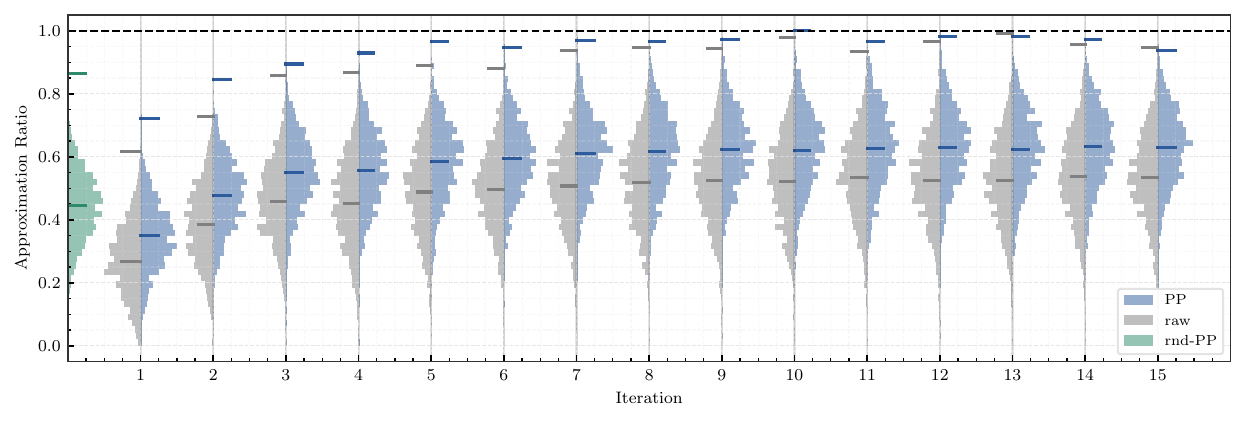}
    \caption{Histograms of the measurement outcomes in terms of approximation ratio from \texttt{ibm\_boston} at a circuit depth of $p=1$ and sample size $M=200$ over the first 15 IWS iterations for problem instance 2. Data from all 10 independent repetitions are aggregated. Raw sampled data are depicted in gray, while post-processed (PP) samples are shown in blue. Horizontal lines denote the mean and the best sampled values for each distribution. The standalone green histogram on the far left establishes a baseline using random sampling combined with post-processing (rnd-PP).}
    \label{fig:hw-hist-trace}
\end{figure*}

\subsection{Results}

\begin{figure*}
    \centering
    \includegraphics[width=\linewidth]{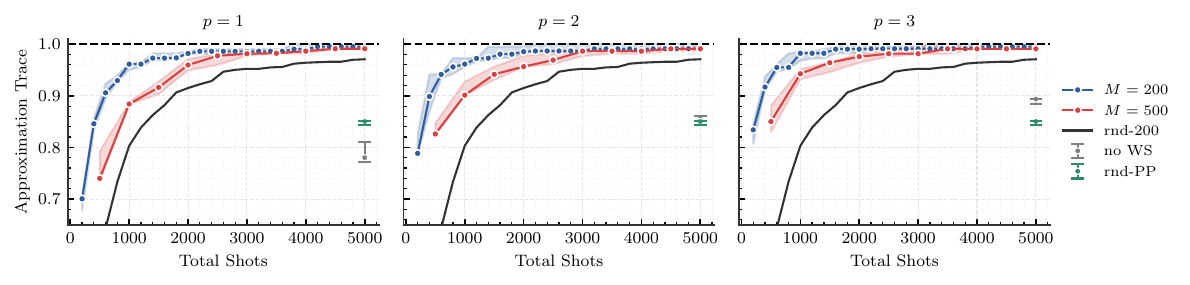}
    \caption{Best objective value found across the 10 independent repetitions as a function of total accumulated shots, plotted for different QAOA depths $p$. Solid lines represent the median over the 5 problem instances, with error bands indicating the interquartile range. The black solid line denotes the IWS heuristic utilizing random sampling with $M=200$. The gray data points the baseline QAOA performance without IWS, evaluated with 5000 shots across 10 repetitions. The green data points represent the rnd-PP baseline evaluated with 50\,000 samples. Note that the random baseline results are identical across all panels.}
    \label{fig:hw-trace}
\end{figure*}

\begin{table*}[]
    \centering
\begin{tabular}{r @{\hspace{2em}} r @{\hspace{2em}} rrrr @{\hspace{2em}} rrrr @{\hspace{2em}} rrrr}
\toprule
 & & \multicolumn{4}{c}{No WS} & \multicolumn{4}{c}{$M = 200$} & \multicolumn{4}{c}{$M = 500$} \\
\cmidrule(r){3-6} \cmidrule(r){7-10} \cmidrule{11-14}
Instance & optimal & rnd & $p=1$ & $p=2$ & $p=3$ & rnd & $p=1$ & $p=2$ & $p=3$ & rnd & $p=1$ & $p=2$ & $p=3$ \\
\midrule
0 & $-22.2$ & $-10.5$ & $-18.1$ & $-19.5$ & $-20.1$ & $-21.5$ & $-21.8$ & $-22.0$ & $-22.0$ & $-21.5$ & $-21.9$ & $-21.9$ & $-22.0$ \\
1 & $-21.7$ & $-10.7$ & $-17.6$ & $-18.7$ & $-18.9$ & $-21.1$ & $-21.6$ & $\mathbf{-21.7}$ & $\mathbf{-21.7}$ & $-20.2$ & $-21.5$ & $-21.4$ & $-21.5$ \\
2 & $-22.7$ & $-11.2$ & $-17.5$ & $-19.2$ & $-20.3$ & $-22.3$ & $\mathbf{-22.7}$ & $-22.5$ & $-22.5$ & $-22.4$ & $-22.5$ & $\mathbf{-22.7}$ & $\mathbf{-22.7}$ \\
3 & $-22.3$ & $-9.9 $ & $-17.4$ & $-19.0$ & $-19.7$ & $-21.4$ & $\mathbf{-22.3}$ & $\mathbf{-22.3}$ & $\mathbf{-22.3}$ & $-22.1$ & $\mathbf{-22.3}$ & $-22.1$ & $-22.2$ \\
4 & $-20.7$ & $-8.5 $ & $-15.5$ & $-17.7$ & $-18.5$ & $-20.1$ & $-20.5$ & $-20.5$ & $-20.6$ & $-20.0$ & $-20.6$ & $-20.6$ & $-20.4$ \\
\bottomrule
\end{tabular}
\caption{Comparison of the best objective values obtained using standard QAOA, IWS-QAOA evaluated on \texttt{ibm\_boston}, and the IWS heuristic with random sampling. Objective values highlighted in bold are optimal solutions. The true optimal values for each instance were determined using SCIP~\cite{SCIPOptSuite10}.\label{tab:best-obj}}
\end{table*}

Optimizing the variational QAOA parameters on hardware is inherently challenging due to noisy measurement outcomes. For this reason, we avoid on-hardware optimization; instead, we reuse parameters optimized via simulation for hardware-tailored problems with 24, 36, and 48 qubits. Specifically, we average the optimal linear schedule parameters obtained across these problem sizes for each depth $p \in \{1, 2, 3\}$. These parameter schedules are summarized in Table~\ref{tab:qaoa-params}.

We execute IWS-QAOA with hyperparameters $\beta_T=15$ and $\epsilon=0.1$ for higher exploitation, testing $M \in \{200, 500\}$. As a baseline, we also evaluate standard QAOA (without IWS) for a total number of $\overline{M} = 5000$ shots. The total quantum resources, in terms of runtime on the QPU, used for IWS-QAOA and standard QAOA are the same. Each algorithmic configuration is evaluated across the 5 problem instances using 10 independent repetitions. Fig.~\ref{fig:hw-hist-trace} presents the measurement histograms for the first 15 IWS iterations at $p=1$ and $M=200$. This figure contrasts the raw sample distributions with their post-processed counterparts. Post-processing (PP) visibly shifts the distributions toward higher approximation ratios, despite altering only roughly 10\% of the measured bits (cf. Fig.~\ref{fig:hist-constrs}). We observe substantial improvements within the first $\sim 5$ iterations. Subsequently, the rate of improvement decelerates, though the ensemble means continue to shift marginally between iterations 5 and 15. Notably, IWS-QAOA surpasses the random sampling with post-processing (rnd-PP) heuristic by the third iteration.

Fig.~\ref{fig:hw-trace} illustrates the best objective value found across all 10 repetitions as a function of the total accumulated shots. The performance is similar across all instances, with both $M=200$ and $M=500$ consistently approaching the optimal solution. Even though $M=500$ utilizes a larger ensemble to compute the next warm-start iteration, the final solution qualities are comparable to those achieved with $M=200$. Compared to the standard QAOA baseline (solid gray line), IWS-QAOA consistently discovers better solutions. While the initial best solutions are similar across varying depths $p$, the algorithm consistently converges to comparable near-optimal final solutions as it progresses. Comparing these results against the classical baselines reveals that, even without IWS, QAOA (with PP) outperforms rnd-PP at $p=3$. Although the IWS heuristic utilizing purely classical random sampling identifies high-quality solutions after a few iterations, it fails to surpass the performance of IWS-QAOA across all evaluated depths.

Table~\ref{tab:best-obj} displays the best objective values discovered by IWS-QAOA and the IWS random sampling heuristic for each instance. IWS-QAOA successfully identifies the optimal solution in instances 1, 2, and 3. Across all runs, $M=200$ finds 6 optimal solutions, whereas $M=500$ discovers only 3. There is no definitive indication of which QAOA depth is strictly superior, as all depths $p$ yield the same number of optimal solutions. However, the total aggregate deviation from the true optimal values is minimized for $M=200$ and $p=3$ at $0.5$.

\section{Conclusion}\label{sec:conclusion}

In this work, we addressed the challenge of solving constrained combinatorial optimization problems using quantum heuristics by introducing a fully warm-started XY-mixer and integrating it into an iterative warm-start framework. Analytically, we proved that the warm-started $\Wp$ state is the unique ground state of our proposed XY-mixer Hamiltonian within the Hamming-weight-1 subspace, ensuring theoretical alignment between the initial state and the mixer. Furthermore, we provided an efficient circuit implementation using two-qubit Pauli rotations.

Our numerical simulations on Max-$k$-Cut and Traveling Salesperson Problem instances demonstrate that IWS-QAOA significantly accelerates the search for optimal solutions compared to standard XY-QAOA. By iteratively updating the probability distribution based on previous samples, the algorithm effectively boosts the probability of sampling the optimal solution, often by orders of magnitude, but it may also diminish $P_\text{opt}$ when stuck in a local minimum if $p$ is small. We observed that while small sample sizes ($M$) allow for faster iterations, they can lead the algorithm into local minima, especially on more complex landscapes, such as larger TSP instances. However, provided the underlying QAOA is sufficiently capable (e.g., using  $p>1$ QAOA layers), IWS-QAOA consistently overcomes these barriers and accelerates the runtime for sampling good, or even optimal, solutions.

Finally, we successfully deployed our approach on IBM's Heron r3 QPU (\texttt{ibm\_boston}) using hardware-tailored problem instances with 144 qubits. Due to inherent NISQ hardware noise, strict adherence to the Hamming-weight constraints cannot be guaranteed, requiring a classical greedy steepest descent post-processing step to fix infeasible assignments. Our results confirm that the combination of IWS-QAOA and post-processing outperforms both standard QAOA and random sampling baselines, successfully identifying optimal solutions on actual quantum hardware.

In conclusion, warm-starting XY-mixers combined with an iterative update strategy presents a highly effective approach for constrained quantum optimization.

\subsection{Future Directions}

The results presented in this work suggest several natural directions for further improving and extending IWS-QAOA. In particular, future work could investigate the following paths forward:

\paragraph*{Improved probability update rules.}
The present implementation updates the warm-start distribution using a Boltzmann-weighted expectation value over measured samples. Future work could explore alternative probability evaluation methods, such as rank-based updates, CVaR-inspired updates~\cite{cadavid2025}, or adaptive momentum-based update methods from classical gradient descent~\cite{kingma2017}. This could increase the algorithm's stability against local minima and reduce the number of iterations required for convergence.

\paragraph*{Combination with genetic search.}
Along similar lines, the iterative update scheme could be embedded into a genetic search algorithm or combined with other classical metaheuristics. For instance, multiple warm-start distributions could be evolved in parallel, recombined, and selected based on the quality of their sampled solutions. The parallel distributions can be incentivized to promote diversity, helping reduce susceptibility to similar local minima across parallel runs. Quantum genetic algorithms have been explored before and may be adapted accordingly~\cite{king2019, gabor2022}.

\paragraph*{Extension to other constraint types.}
While this work focuses on Hamming-weight-1 constraints, an important direction is to apply IWS to broader classes of constraints. Since XY-mixers are Hamming-weight-preserving, the logical next step is to investigate how to include higher Hamming-weight constraints into the same scheme. Furthermore, different warm-starting Hamiltonians need to be investigated for other classes of constraint-preserving mixer constructions~\cite{hadfield2019, fuchs2022}. For instance, Grover mixers~\cite{bartschi2020} can be warm-started using the Amplitude-Amplification Mixers from Ref.~\cite{christiansen2025}. Such extensions would broaden the range of constrained optimization problems that the proposed algorithm can address.

\paragraph*{XY-Mixer Improvements.}
The implementation considered here prioritizes hardware efficiency. Future work could study the impact of larger Trotter numbers, higher-order decompositions, or alternative, well-suited topologies. These variants may more accurately approximate the intended mixer Hamiltonian and could improve performance, albeit at the cost of increased circuit depth.

\paragraph*{Initialization from stronger classical methods.}
In this work, the iterative procedure does not require a problem-specific classical initialization. Nevertheless, it seems worthwhile to investigate how the method performs when initialized using classical approaches, beyond SDP rounding, as in Ref.~\cite{he2026}. Alternatively, initializing for classical heuristics, such as simulated annealing, could also be a potential route worth considering. Such initializations could provide useful prior information while the iterative quantum-classical update scheme retains the ability to further refine the search distribution. As a consequence, we expect fewer iterations to solve the problem to optimality.

\paragraph*{Improved parameter optimization.}
The numerical simulations in this work use a linear QAOA parameter schedule. Future work could investigate more advanced scheduling strategies, such as Fourier parameterizations~\cite{zhou2020}, interpolation schemes~\cite{apte2026} and parameter transfer between related instances~\cite{shaydulin2023}. Additionally, shot-friendly optimization techniques are required for parameter optimization on QPUs~\cite{hao2025}. These approaches may improve performance and reduce the overhead of classical optimization.

\paragraph*{Error mitigation and correction.}
Finally, the hardware experiments in this work rely on classical post-processing to repair infeasible samples caused by noise. Advanced error-mitigation~\cite{nation2021} or optimization-focused error-correction~\cite{he2025} techniques~ could reduce this dependence and make the quantum part of the algorithm more directly effective. As hardware fidelity improves, it will also be important to systematically analyze the scalability of the proposed hybrid approach on larger and less hardware-tailored problem instances.

\section*{Acknowledgments}
This work was supported by the German Federal Ministry of Research, Technology, and Space (BMFTR) under the funding program
“Förderprogramm Quantentechnologien -- von den Grundlagen zum Markt” (funding program quantum technologies -- from basic research to market), project QuCUN, 13N16199. We acknowledge the use of IBM Quantum Credits via the IBM Quantum Startups Program for this work. The views expressed are those of the authors and do not reflect the official policy or position of IBM or the IBM Quantum Platform team.
\section*{Data Availability}
Data from the experiments in the manuscript is available in the repositoriy~\href{https://github.com/aqarios/warm-start-xy-data}{github.com/aqarios/warm-start-xy-data}.
\bibliographystyle{myapsrev4-2}
\bibliography{references}
\appendix
\clearpage
\onecolumngrid

\section{State Preparation of the Biased W-State}\label{app:wstate}
In this section, we discuss the circuits required to construct the biased $\ket{W_P}$ state.

\paragraph{Linear Synthesis}
The standard $\ket{W}$ state for $k$ qubits is constructed starting from the state $\ket{e_1} = \ket{10\cdots0}$. Following Ref.~\cite{cruz2019}, we apply a sequence of gates $B_{ij}(q) = \textsc{cnot}_{ji} C(R_Y(2 \arccos \sqrt{q}))_{ij}$, where $C(\cdot)$ denotes a controlled application. These gates have the following properties:
\begin{align}
    B(q)\ket{00} = \ket{00} \quad\text{and}\quad B(q)\ket{10}= \sqrt{q}\ket{10} + \sqrt{1-q}\ket{01},
\end{align}
effectively acting as parameterized \textsc{swap} gates. Applying $B$ in a linear chain from the initialized qubit yields the equal superposition state:
\begin{align}
    \prod_{i=1}^{k-1} B_{i,i+1}\left(\frac{1}{k-i+1}\right) \ket{e_1} = B_{k-1,k}(1/2)\cdots B_{1,2}(1/k)\ket{e_1} = \frac{1}{\sqrt{k}}\sum_i \ket{e_i} = \ket{W}.
\end{align}
The resulting state preparation circuit has a depth of $O(k)$, as each operator depends on the preceding one.

This scheme can be generalized to the biased state $\ket{W_P}$ by adjusting the parameters of $B$ to match the probability distribution $P$:
\begin{align}\label{eq:wp-state-preparation}
    \ket{W_P} = \prod_{i=1}^{k-1} B_{i,i+1}\left(\frac{P_i}{1-\sum_{j < i} P_{j}}\right)\ket{e_1} = \sum_i \sqrt{P_i}\ket{e_i}.
\end{align}

\paragraph*{Example} 
Given a probability distribution for 4 bits $P_1,\dots,P_4$, we iteratively apply the gates from~\eqref{eq:wp-state-preparation}:
\begin{align*}
\ket{e_1} &\xrightarrow{B_{1,2}(P_1)} \sqrt{P_1}\ket{e_1} + \sqrt{1-P_1}\ket{e_2} \\
&\xrightarrow{B_{2,3}(P_2 / (1 - P_1))} \sqrt{P_1}\ket{e_1} + \sqrt{P_2}\ket{e_2} + \sqrt{1 - P_1 - P_2}\ket{e_3} \\
&\xrightarrow{B_{3,4}(P_3 / (1 - P_1 - P_2))} \sqrt{P_1}\ket{e_1} + \sqrt{P_2}\ket{e_2} + \sqrt{P_3}\ket{e_3} + \sqrt{1 - P_1 - P_2 - P_3}\ket{e_4} = \ket{W_P},
\end{align*}
where the final step utilizes the normalization $\sum_i P_i = 1$.

\paragraph{Logarithmic Synthesis} 
As shown in Ref.~\cite{cruz2019}, the $B$ gates can be arranged in a binary tree structure, reducing the circuit depth to $O(\log k)$. However, this implementation requires higher inter-connectivity in the QPU topology.

\paragraph{NISQ-Friendly Synthesis} 
The initial control of the $R_Y$ gate is redundant because the control qubit is initialized in $\ket{1}$; we thus replace $C(R_Y)$ with a standard $R_Y$ rotation. By placing the initial qubit at the center $\lfloor k / 2 \rfloor$ and applying $B$ gates outward in two branches, we halve the depth of the linear synthesis while maintaining NISQ-friendly linear connectivity. This optimized approach is employed in our hardware experiments.

\section{Ground State Proofs of Warm-Started Mixers}\label{sec:proofs-gs}

\begin{proof}[Proof of Corollary~\ref{cor:1}]
The operator $\mathcal{H}^G_P$ preserves the Hamming-weight 1 subspace because $[\mathcal{H}_{ij}(q), (I-Z)_i + (I-Z)_j] = 0$, as established in the proof of Proposition~\ref{th:Hp-Wp-groundstate}. Since $G$ is connected, there exists $n \in \mathbb{N}$ such that $\bra{e_i}(\mathcal{H}^G_P)^n\ket{e_j} \neq 0$ for all $i \neq j$, ensuring a path exists between any two qubits in the mixer topology. Given that the off-diagonal entries $\bra{e_i}\mathcal{H}^G_P\ket{e_j}$ are non-positive, the Perron-Frobenius theorem implies that the unique ground state must have strictly positive real entries, a property satisfied by $\ket{W_P}$.

It remains to show that $\ket{W_P}$ is an eigenstate. Since $G$ is regular with degree $d = \Delta(G)$, we have:
\begin{align}
    \mathcal{H}^G_P\ket{W_P} &= \frac{1}{d} \sum_{i,j}\mathcal{H}(q_{ij}) \ket{W_P} \\
    &= \frac{1}{d} \sum_{i = 1}^k \sum_{j \in N_G(i)} \left[ \sqrt{P_i} \frac{P_j - P_i}{P_i + P_j} \ket{e_i} - \frac{2P_i\sqrt{P_j}}{P_i + P_j} \ket{e_j} \right] \\ 
    &= -\frac{1}{d} \sum_{i=1}^k \sum_{j \in N_G(i)} \sqrt{P_i}\ket{e_i} = -\ket{W_P},
\end{align}
where $N_G(i)$ denotes the neighborhood of node $i$ with $|N_G(i)| = d$. Consequently, $\ket{W_P}$ is the unique ground state of $\mathcal{H}_P^G$ within the Hamming-weight 1 subspace, corresponding to an energy eigenvalue of $-1$.
\end{proof}

\begin{proof}[Proof of Corollary~\ref{cor:app}]
Since $[\ket{e_i}\bra{e_i}, Z_i] = 0$, the operator $\mathcal{H}_P^{G}$ preserves the Hamming-weight 1 subspace. As $G$ is connected, the Perron-Frobenius theorem implies that $\ket{W_P}$ is the unique ground state provided it is an eigenstate. We evaluate the action of $\mathcal{H}_P^G$ as follows:
\begin{align}
    \mathcal{H}^{G}_P\ket{W_P} &= \frac{1}{\Delta(G)}\sum_{i = 1}^k \sum_{j \in N_G(i)} \left[ \sqrt{P_i} \frac{P_j - P_i}{P_i + P_j} \ket{e_i} - \frac{2P_i\sqrt{P_j}}{P_i + P_j} \ket{e_j} \right] + \frac{1}{\Delta(G)} \sum_{i=1}^{k} (\deg(i) - \Delta(G)) \sqrt{P_i} \ket{e_i}\\ 
    &= \frac{1}{\Delta(G)} \sum_{i=1}^k \left[ -\deg(i) \sqrt{P_i}\ket{e_i} + (\deg(i) - \Delta(G))\sqrt{P_i}\ket{e_i} \right] = -\ket{W_P}.
\end{align}
Thus, $\ket{W_P}$ is the unique ground state of $\mathcal{H}_P^G$ in the Hamming-weight 1 sector with energy $-1$.
\end{proof}

\begin{proof}[Proof of Corollary~\ref{cor:2}]
Regardless of the connectivity of the mixer topology, $\mathcal{H}_P^{G}$ preserves the total Hamming weight because each constituent term commutes with the number operator, as shown previously. Furthermore, we know $\Delta(G) \geq 1$ since $E \neq \emptyset$.

When $G$ is disconnected, the Hamiltonian decomposes into separable sectors that can be diagonalized independently. Each sector corresponds either to a connected subgraph $G'(V', E') \subset G$ with $|V'| = k'$, or to an isolated node with a local Hamiltonian $H_l = -\ket{e_l}\bra{e_l}$ and a trivial ground state $\ket{e_l}$. From Corollary~\ref{cor:app}, we know that $\ket{W_{P'}} = \sum_{i\in V'} \sqrt{P'_i} \ket{e_i}$ (where $P'_i = P_i / \sum_{j\in V'} P_j$) is the unique ground state of $\mathcal{H}^{G'}_{P'}$ with energy $-1$. Rescaling and shifting the sector Hamiltonian as
\begin{align}
    \hat{\mathcal{H}}^{G'}_{P'} = \frac{\Delta(G')}{\Delta(G)}\mathcal{H}^{G'}_{P'} - \frac{\Delta(G) - \Delta(G')}{\Delta(G)}\sum_{i\in V'} \ket{e_i}\bra{e_i}
\end{align}
leaves the ground-state properties and the eigenvalue $-1$ intact. 

The global Hamiltonian $\mathcal{H}^G_P = \sum_{G'}\hat{\mathcal{H}}_{P'}^{G'} + \sum_l H_l$ possesses a degenerate ground-state manifold with energy $-1$, spanned by the individual sector ground states $\{\ket{W_{P'}}\}$ and $\{\ket{e_l}\}$. While a product of these states would exit the Hamming-weight-1 subspace, any normalized superposition remains within it. By choosing the specific superposition 
\begin{align}
    \sum_{V' \in \mathcal{S}} \sqrt{\sum_{i\in V'} P_i} \ket{W_{P'}} + \sum_l \sqrt{P_l}\ket{e_l} = \ket{W_P},
\end{align}
we confirm that $\ket{W_P}$ is a valid ground state of $\mathcal{H}_P^{G}$ with energy $-1$ for any mixer topology $G$ within the Hamming-weight-1 subspace.
\end{proof}

\section{Circuit Implementation of the Warm-Start XY-Block}\label{sec:proof-rotation}

\begin{proof}
We demonstrate that the following equality holds:
\begin{align*}
e^{-i\beta \mathcal{H}(q)} = (R_Z(\phi_1) \otimes I) U_{XY}(\phi_2) (I \otimes R_Z(-\phi_1))
\end{align*}
where the rotation angles are defined as:
\begin{align*}
\phi_1 = \mathrm{arctan2}\left((1-2q)\sin\beta, \cos\beta\right), \quad \phi_2 = \arcsin\left(2\sqrt{q(1-q)} \sin\beta\right).
\end{align*}
We focus the analysis on the $\{\ket{01}, \ket{10}\}$ subspace, as the identity holds trivially for the $\{\ket{00}, \ket{11}\}$ sector.

Expanding the left-hand side within the $\{\ket{01}, \ket{10}\}$ subspace yields:
\begin{align}\label{eq:rotproof-lhs}
e^{-i\beta \mathcal{H}(q)} = \begin{pmatrix} \cos\beta - i (1-2q)\sin\beta & 2i\sqrt{q(1-q)} \sin\beta \\ 2i\sqrt{q(1-q)} \sin \beta & \cos\beta + i (1-2q)\sin\beta \end{pmatrix}.
\end{align}

Similarly, the right-hand side in the same subspace is given by:
\begin{align}\label{eq:rotproof-rhs}
(R_Z(\phi_1) \otimes I) U_{XY}(\phi_2) (I \otimes R_Z(-\phi_1)) =
    \begin{pmatrix}
        e^{-i\phi_1}\cos{\phi_2} & i\sin{\phi_2} \\ i \sin{\phi_2} & e^{i\phi_1} \cos{\phi_2}
    \end{pmatrix}.
\end{align}
Substituting $\phi_2$ into the off-diagonal elements of Eq.~\eqref{eq:rotproof-rhs} directly recovers the off-diagonals of Eq.~\eqref{eq:rotproof-lhs}. 
Utilizing the identities $\cos(\arcsin x) = \sqrt{1-x^2}$ and $e^{\pm i \mathrm{arctan2}(y, x)} = (x \pm iy) / \sqrt{x^2 + y^2}$, we expand the diagonal terms:
\begin{align}
    e^{\mp i\phi_1}\cos{\phi_2} &= \frac{\cos{\beta} \mp i(1 - 2q) \sin{\beta}} {\sqrt{\cos^2\beta + \sin^2 \beta (1-2q)^2}} \sqrt{1 - 4q(1-q) \sin^2\beta} \nonumber \\
    &= \cos{\beta} \mp i(1 - 2q)\sin\beta.
\end{align}
The final simplification holds because $\cos^2\beta + (1-2q)^2\sin^2\beta = 1 - 4q(1-q)\sin^2\beta$, which cancels the denominator. Thus, the left-hand side~\eqref{eq:rotproof-lhs} and the right-hand side~\eqref{eq:rotproof-rhs} are identical.
\end{proof}

\section{Re-Optimization versus Fixed Parameters}
\label{app:reopt}

\begin{figure}
\centering
    \subfloat[Approximation ratio]{\includegraphics[width=0.95\linewidth]{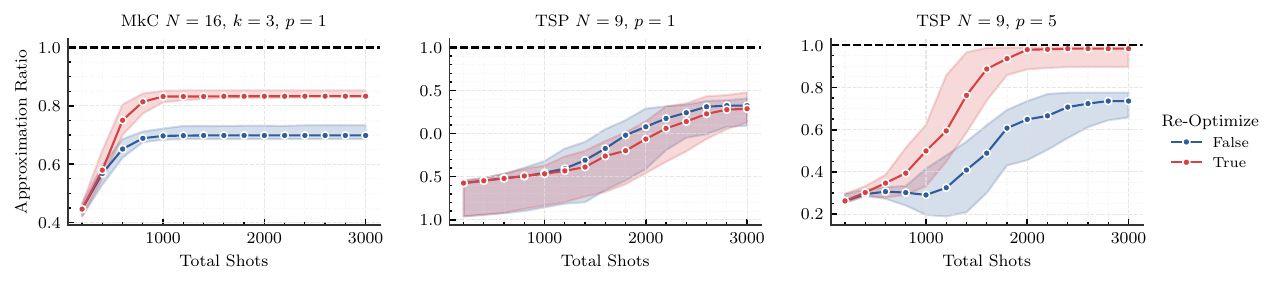}\label{fig:reopt-ar}}\\
    \subfloat[Approximation trace]{\includegraphics[width=0.95\linewidth]{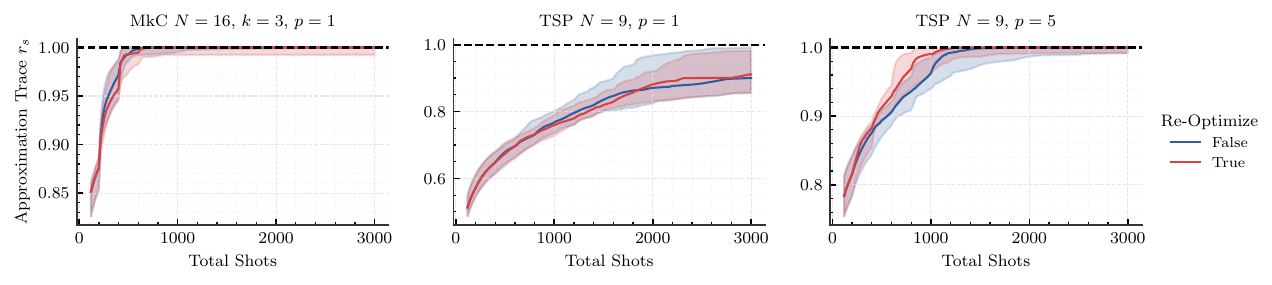}\label{fig:reopt-at}}
\caption{Comparison between fixing parameters at the first iteration and re-optimizing parameters every IWS iteration. Results showing the approximation ratio per iteration (a) and approximation trace in relation to shots (b). Data points show the median of 10 repetitions over 5 instances with error bars marking the IQR.\label{fig:reopt}}
\end{figure}

In Sec.~\ref{sec:iws-alg}, we argued that aligning the warm-started initial state with the warm-started XY-mixer keeps the optimal QAOA parameter region approximately invariant under iterative updates of the bias~$P^{(t)}$, allowing us to optimize the variational parameters only once at the start of Algorithm~\ref{alg:iws-qaoa}. The qualitative landscape analysis in Sec.~\ref{sec:parameter-validation} and Fig.~\ref{fig:landscape} supports this assumption for a representative TSP instance. To verify it quantitatively across instance classes, we compare two variants of IWS-QAOA:
\begin{itemize}
    \item The linear schedule~\eqref{eq:linear-schedule} is optimized \emph{once} against the uniform initial bias $P^{(0)}$ and held fixed throughout all IWS iterations, as described in Algorithm~\ref{alg:iws-qaoa}.
    \item The linear schedule is re-optimized at every iteration using BFGS, reusing the parameters of the previous iteration as initial points for the optimizer.
\end{itemize}
We evaluate both variants on 5 instances each of MkC ($k=3$, $N=16$) and TSP ($N=9$), with the TSP instances tested at both $p=1$ and $p=5$. Each configuration is repeated 10 times with $M=200$, $\varepsilon = 0.2$, and $\beta_T = 15$. Fig.~\ref{fig:reopt} reports the resulting approximation ratio and approximation trace.

Since the bias deforms the cost landscape at every iteration (as shown in Fig.~\ref{fig:landscape}), re-optimization is, in principle, expected to improve performance. Panel~\subref{fig:reopt-ar} confirms this: the per-iteration approximation ratio is higher under re-optimization for MkC and TSP at $p=5$, indicating that the QAOA landscape does shift with warm-starting and is therefore not strictly invariant. No improvement can be observed in the TSP $p=1$ instance, however. Overall, the magnitude of this shift is modest and does not translate into a meaningful improvement in the approximation trace (panel~\subref{fig:reopt-at}). The trace curves nearly coincide across all instances, with re-optimization marginally outperforming the fixed schedule only on TSP at $p=5$. The remaining cases show no noticeable difference in the trace.

From a practical standpoint, re-optimization incurs substantial classical optimization overhead that is disproportionate to the gain measured here: the marginal loss from reusing parameters is offset by only a handful of additional IWS shots, whereas a full optimizer run per iteration is considerably more expensive, especially on hardware. Therefore, we retain the optimize-once strategy in IWS-QAOA.

\section{Hardware Instance Qubit-Triplet Selection Problem}\label{sec:hw-optimization}

To generate the hardware-tailored problem instances, we first find all connected three-qubit paths in the hardware coupling map $G(V, E)$ with associated edge errors $w_{uv}$, and collect them into the triplet set $\mathcal{T}$. Furthermore, we group all pairs of non-overlapping triplets $t_1, t_2 \in \mathcal{T}, t_1 \cap t_2 = \emptyset$ that are connected by a coupler, $\exists (u, v) \in E: (u \in t_1 \wedge v\in t_2) \vee (u \in t_2 \wedge v\in t_1)$, into the interconnection set $\mathcal{I} \subset \mathcal{T}^2$. We assign each triplet $t \in \mathcal{T}$ a binary variable $x_t \in \{0, 1\}$ and each pair $(t,l) \in \mathcal{I}$ a binary variable $y_{tl}$, and formulate the following linear binary optimization problem:
\begin{align}
\begin{aligned}
    \max_{x,y}& \sum_{(t,l)\in \mathcal{I}} \left(1 - \frac{w_{tl}}{2W}\right)y_{t l} - \sum_{t\in\mathcal{T}} x_t\sum_{(i, j)\in (t^2 \cap E)}\frac{w_{ij}}{2W},\\
    \text{such that:}&\\
    & x_t + x_l \geq 2 y_{tl} \quad \forall (t,l) \in \mathcal{I}\\
    & x_t + x_l \leq 1 \quad\quad \forall t,l \in \mathcal{T}, t \neq l, t\cap l\neq \emptyset \\
    & \sum_{t\in \mathcal{T}}x_t = N
\end{aligned}
\end{align}
where $W = \max_{e\in E}w_e$, and $N$ is the number of triplets to select. Here, we associate $w_{tl}$ with the edge weight of the only existing edge between the connected triplets. Due to the heavy-hex topology, there is at most one edge per qubit-triplet pair. The objective maximizes the selected interconnections between triplets while minimizing the aggregated error of all used couplers, weighted by $1/{2W}$. We solve this optimization problem using the CP-SAT solver from Google OR-Tools~\cite{cpsatlp} to arrive at the ideal selection of qubit triplets, as visualized in Fig.~\ref{fig:coupling_map}.

\section{Isolation of the Post-Processing Effect}
\label{app:postprocessing}

\begin{figure}
\subfloat[Post-Selection]{\includegraphics[width=0.35\textwidth]{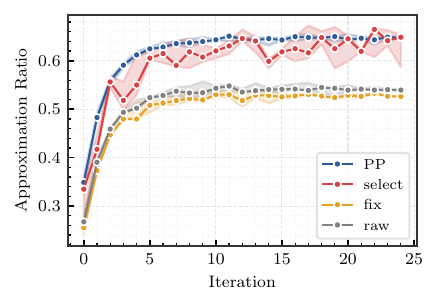}\label{fig:opt_postselect}}
\subfloat[Comparison of Heuristics]{\includegraphics[width=0.4\textwidth]{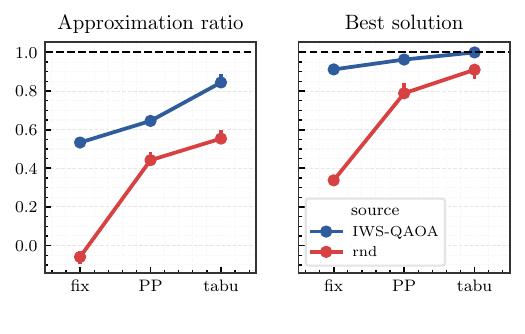}\label{fig:opt_comparison}}
\subfloat[Hamming Distance]{\includegraphics[width=0.25\textwidth]{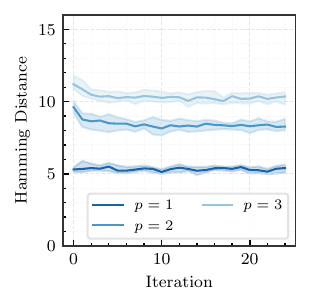}\label{fig:opt_hamming}}
\caption{Isolation of the post-processing contribution for the 144-bit instance. (a) Median approximation ratio per IWS iteration at $p=1$ and $M=200$ for the four sample-handling methods (raw, fix, select, PP). (b) Approximation ratio (left) and best solution found (right) for the three post-processors of increasing strength (fix, PP, tabu), applied to IWS-QAOA samples at iteration 10 and to random bit-strings (rnd). (c) Hamming distance between the raw and post-processed samples per IWS iteration for $p\in\{1,2,3\}$. All data points show the median with error bars marking the IQR.}
\label{fig:postprocessing}
\end{figure}

To isolate the effect of the classical post-processing on the hardware results, we proceed in three steps. We first study post-selecting only the raw feasible samples, which is possible at $p=1$, where the feasibility rate still allows some analysis. Next, we compare post-processing heuristics of increasing strength against one another for both quantum and random sample sources. Finally, we examine the number of bit-flips that the optimizer applies to each sample.

\subsection{Post-Selecting Instead of Processing}\label{app:selec}
The greedy descent of Sec.~\ref{sec:postprocessing} mainly recovers feasibility from measured infeasible samples by locally optimizing the cost of the quadratically penalized objective. To investigate its effect, we compare four modes of handling the raw measurements: every sample is scored directly through the non-penalized cost function (raw), infeasible one-hot assignments are randomly fixed (fix), only samples that satisfy all one-hot constraints are selected without application of a post-processor (select), and the greedy steepest-descent repair used in the main text (PP). Fig.~\ref{fig:opt_postselect} shows the median approximation ratio per IWS iteration at $p=1$ and $M=200$. This analysis is only statistically meaningful at $p=1$. At larger depths, the number of fully feasible samples collapses (cf.~Sec.~\ref{sec:postprocessing}), which is precisely the motivation for the repair step.

Crucially, the post-select curve tracks the PP curve closely. Since post-selection is an iteration-independent operation that cannot move a sample toward lower cost, this can only originate from IWS-QAOA concentrating the feasible quantum distribution on low-cost states. The results indicate that the cost-optimizing component of the greedy repair contributes only marginally to the final solution quality and that its principal role is the recovery of feasibility. The \emph{raw} and \emph{fix} curves improve in parallel but lie consistently below the other two. Notably, the \emph{fix} curve lies slightly below the \emph{raw} curve, meaning that randomly fixing constraints decreases QAOA performance, as many unfavorable configurations are selected by chance.

\subsection{Comparing Post-Processors}
Next, we analyze the performance of three post-processors on raw samples from the 10th iteration of the IWS-QAOA runs and on completely random bit-strings. The three heuristics are of increasing strength: PP and fix, as discussed in Sec.~\ref{app:selec}, are compared against a tabu search with a small tenure size of 1. Increasing the tabu tenure would solve the 144-bit instance to optimality from arbitrary input. Fig.~\ref{fig:opt_comparison} reports the mean approximation ratio and the best solution found.

Across all three post-processors, IWS-QAOA samples yield clearly better solutions than random ones. With \emph{fix}, even under a sub-optimal repair that does not optimize cost, IWS-QAOA attains a higher approximation ratio and best solution than random sampling combined with PP. The advantage is therefore present before any cost optimization is applied and is a property of the quantum distribution itself. Likewise, the stronger tabu post-processor benefits from the better initial samples provided by IWS-QAOA. We can therefore conclude that the performance increase is attributable to the IWS-QAOA sampler and not to the post-optimization step.

\subsection{Bit-Flip Analysis}\label{app:hamming}
Finally, we quantify how much the greedy repair modifies each sample by measuring the Hamming distance between the raw measurement and its post-processed counterpart. Fig.~\ref{fig:opt_hamming} shows this distance per IWS iteration for $p\in\{1,2,3\}$. At $p=1$, the optimizer flips on average only $\sim5$ of the $144$ bits ($\approx3.5\%$), rising to $\sim8$ at $p=2$ and $\sim10$ at $p=3$ as the noise-induced violation rate grows with depth. These values coincide closely with the expected number of violated constraints $N f(p)$, which evaluates to $5.3$, $8.6$, and $10.6$ for the failure rates $f(p)$ reported in Sec.~\ref{sec:postprocessing}. This confirms that the repair performs essentially a single flip per violated constraint, consistent with the analytical argument that post-processing freezes every already-feasible one-hot constraint and acts only on the small set of violating variables. For rnd-PP, we observe 36 bit-flips on average, corresponding to $25\%$ of the bit-string being modified.

\end{document}